%% file: example.tex
\documentclass[submission, copyright, creativecommons]{eptcs}
\providecommand{\event}{GandALF 2023} 

\usepackage{iftex}

\ifpdf
  \usepackage{underscore}         
  \usepackage[T1]{fontenc}        
\else
  \usepackage{breakurl}           
\fi

\title{(Un)Decidability Bounds of the Synthesis Problem \\ for Petri Games}
\author{Paul Hannibal
\institute{University of Oldenburg, \\ Lower Saxony, Germany }
\email{paul.jonathan.hannibal1@uni-oldenburg.de}
}
\def\titlerunning{(Un)Decidability Bounds for Petri Games}
\def\authorrunning{P. Hannibal}

\usepackage{amsmath, amssymb,amsthm, amsfonts}
\usepackage{float}

\usepackage{tikz}
\usetikzlibrary{automata,calc,shapes.geometric,positioning,trees,fit,chains,matrix,scopes,arrows, petri, backgrounds}

\tikzstyle{eplace}=[circle,thick,draw=black!75,fill=white!20,minimum size=5mm]
\tikzstyle{splace}=[circle,thick,draw=black!75,fill=black!25,minimum size=5mm]
\tikzstyle{red place}=[place,draw=red!75,fill=red!20]
\tikzstyle{transition}=[rectangle,thick,draw=blue!75,
fill=blue!20,minimum size=3.2mm]

\tikzstyle{every label}=[black]

\newcommand{\future}{\mathit{Fut}}

\newcommand{\unf}{\textit{unf}}

\newcommand{\transitions}{\mathcal{T}}
\newcommand{\places}{\mathcal{P}}
\newcommand{\preset}{\mathit{pre}}
\newcommand{\postset}{\mathit{post}}
\newcommand{\initmarking}{\mathit{In}}
\newcommand{\flowrelation}{\mathcal{F}}
\newcommand{\lkm}{\mathit{lkm}}
\newcommand{\lkc}{\mathit{lkc}}

\newcommand{\reachablemarkings}{\mathcal{R}}
\newcommand{\badmarkings}{\mathcal{B}}

\newcommand{\powersetnf}[1]{2_{\mathit{nf}}^{#1}}
\newcommand{\powerset}[1]{ 2^{#1}}
\newcommand{\powersetf}[1]{ 2_{\mathit{f}}^{#1}}
\newcommand{\squares}{\mathit{DP}}
\newcommand{\lowertriangles}{\mathit{HP}}
\newcommand{\uppertriangles}{\mathit{VP}}
\newcommand{\colors}{\mathit{C}}
\newcommand{\coloringproblem}{\mathit{CP}}
\newcommand{\synchronisationsegment}{\mathit{Seg}}

\newcommand{\partialrepetition}{\mathit{prc}}
\newcommand{\loopcuts}{\partialrepetition_{\mathit{loop}}}
\newcommand{\synchrequivcuts}{\mathit{sqc}}

\theoremstyle{plain}
\newtheorem{theorem}{Theorem}[section]
\newtheorem{lemma}[theorem]{Lemma}

\theoremstyle{definition}
\newtheorem{definition}{Definition}[section]

\theoremstyle{remark}

\usepackage{hyperref}

\begin{document}
\maketitle

\begin{abstract}
Petri games are a multi-player game model for the automatic synthesis of
distributed systems,
where the players are represented as tokens on a Petri net and are grouped
into environment players
and system players.
As long as the players move in independent parts of the net, they do not
know of each other;
when they synchronize at a joint transition, each player gets informed
of the
entire causal history of the other players.

We show that the synthesis problem for two-player Petri games under a global safety condition is NP-complete and it can be solved within a non-deterministic exponential upper bound in the case of up to 4 players. Furthermore, we show the undecidability of the synthesis problem for Petri games with at least 6 players under a local safety condition. 
\end{abstract}

\section{Introduction}
\input{sections/1-introduction-revised.tex}

\section{Foundations }
\label{sec-foundation}
\input{sections/2-foundations.tex}

\section{Decidability Results}
\label{sec-decidability}
\input{sections/4-decidability-results.tex}

\section{Undecidability Result}
\label{sec-undecidability}
\input{sections/3-undecidability-results.tex}

\section{Conclusions}
\label{sec-related-work}
\input{sections/5-related-work-and-conclusions.tex}
\bibliographystyle{eptcs}
\bibliography{refs}

\end{document}

%% file: sections/1-introduction-revised.tex
A Petri game is a model for distributed, reactive systems. It is played on a Petri net where each place is either a system place or an environment place. The tokens on system places are system players and they control which transitions to take next. The tokens on environment places are uncontrollable environment players. 
Essential for Petri games is the informedness of the players. 
As long as the players move in independent parts of the net, they do not know of each other;
when taking a joint transition they exchange information about their complete causal history. 

A winning strategy of the system players reacts to all options of the environment players while satisfying a winning condition. Thereby, a decision of a system player is based on its causal history, which grows infinitely in a Petri net with loops. Different causal histories allow different decisions. The \emph{synthesis problem} asks whether there is a winning strategy of the system players.
There have been several positive results on deciding the synthesis problem for Petri games, obtained by restricting the number of players \cite{FINKBEINER2017181,DBLP:journals/corr/abs-1710-05368} or restricting the concurrent behaviour \cite{DBLP:conf/apn/HannibalO22}.
Also, an approach to bounded synthesis has been proposed \cite{DBLP:conf/birthday/Finkbeiner15}. 
These papers considered  as winning conditions either local safety conditions where some `bad' places must be avoided or 
global safety conditions requiring that some sets of places are considered as `bad' markings that must not be reached simultaneously.

Petri games are related to the model of control games played by multiple processes on Zielonka automata. 
These games are also based on the causal memory of their processes.
A control game is a composition of local processes. The processes communicate via shared actions that are either controllable or uncontrollable. A strategy consists of one local controller for each process that can restrict controllable actions based on the causal past of the process. As in Petri games, a strategy must take into account all uncontrollable behaviour in order to win. Unlike Petri games, one of the most common winning conditions is a local termination condition.
A formal relationship of the two models has been presented in \cite{DBLP:conf/concur/BeutnerFH19},
where translations from Petri games into control games and back have been presented
such that there is a weak bisimulation between the winning strategies of the two models.
When translating a control game into a Petri game, the processes are turned one-by-one
into players. These players switch between system and environment players. The winning condition stays the same. The number of places (or states) blows up exponentially when a game is translated in either direction.

Decidability results for control games have been obtained by restricting the communication architecture \cite{DBLP:conf/icalp/Muscholl15,DBLP:conf/icalp/GenestGMW13} or restricting the concurrent behaviour \cite{DBLP:conf/concur/MadhusudanT02,DBLP:conf/fsttcs/MadhusudanTY05}. Another class of decidable control games are decomposable games \cite{DBLP:conf/fsttcs/Gimbert17} that come with a local termination condition. Decomposable games are decidable with up to 4 players \cite{DBLP:conf/fsttcs/Gimbert17}.
In Sec. \ref{sec-decidability}, we show  that the synthesis problem for Petri games under a global safety condition with up to 4 players is decidable within an exponential upper time bound, and for two-player Petri games this problem is NP-complete.

In \cite{DBLP:journals/lmcs/Gimbert22}, it has been shown that the synthesis problem is undecidable
for control games  with at least 6 processes under a local termination condition. 
This result, together with the translation into Petri games in \cite{DBLP:conf/concur/BeutnerFH19},  
implies that there are 6-player Petri games that are equipped with a local termination condition for which the synthesis problem is undecidable. 
In Sec. \ref{sec-undecidability}, we show in a direct proof that the synthesis problem for Petri games with 6 players is also undecidable under a local safety condition. 

The synthesis problem is of great interest because it automates the error-prone implementation process while delivering implementations that are correct by construction. It was first introduced in \cite{A.Church-reactive-systems}. Pnueli and Rosner introduced a setting of synchronous processes that communicate via shared variables~\cite{DBLP:conf/focs/PnueliR90}. For a single process, this setting is known to be decidable \cite{DBLP:journals/corr/abs-1711-10636,DBLP:conf/icalp/PnueliR89}. For multiple processes, the  setting of Pnueli and Rosner is known to be undecidable \cite{DBLP:conf/focs/PnueliR90}. 
There have been positive decidability results on specific architectures with multiple processes, including pipelines \cite{Rosnerpipelines}, rings \cite{Kupferman2001389}, and acyclic architectures \cite{synthesis-FinkbeinerS05}. However, all the positive results for multiple processes have non-elementary complexity. 

%% file: sections/2-foundations.tex
In this section, we define branching processes and unfoldings as in \cite{DBLP:journals/acta/Engelfriet91}. Also, we define Petri games and their winning strategies as in \cite{FINKBEINER2017181}.

Some notation:	the \emph{power set} of a set $A$ is denoted by $\powerset A = \{ B \mid B \subseteq A \}$, 
the \emph{set of nonempty finite subsets} of $A$ by $\powersetnf A = \{ B \mid B \subseteq A \land B \neq \emptyset  \land B \text{ is finite} \}$, 
and the \emph{set of finite subsets} of $A$ by~$\powersetf A$.

A \emph{Petri net} or simply \emph{net} is a structure $N = \mathcal{(P, T }, \mathit{pre}, \mathit{post}, \mathit{In})$, where $\mathcal{P}$ is the (possibly infinite) set of \emph{places}, $\mathcal{T}$ is the (possibly infinite) set of \emph{transitions}, $\mathit{pre}$ and $\postset$ are \emph{flow mappings}, 
$\mathit{In} \subseteq \mathcal{P}$ is the \emph{initial marking}, and the following properties hold: $\mathcal{P} \cap \mathcal{T} = \emptyset$, $ \preset: \transitions \to \powersetnf\places $, 	$ \postset: \transitions \to \powersetf\places$.
A Petri net is called finite if $\places \cup \transitions$ is a finite set.
The flow mappings $\preset$ and $\postset$ are extended to places as usual: $\forall p \in \places: \preset(p) = \{t \in \transitions \mid p \in \postset(t) \}$ and $\forall p \in \places: \postset(p) = \{t \in \transitions \mid p \in \preset(t) \}$.
A \emph{marking} $M$ of a Petri net $N$ is a \emph{multiset}  over $\mathcal{P}$. In particular, $\initmarking$ is a marking. 
By convention, a net named $N$ has the components $ \mathcal{(P, T }, \mathit{pre}, \mathit{post}, \mathit{In})$,
and analogously for net with decorated names like $N_0, N_1, N_2$, where the components are equally decorated.

A transition $t \in \mathcal{T}$ is \emph{enabled} at marking $M$ if $\preset(t) \subseteq M$. If $t$ is enabled, the transition $t$ can be \emph{fired}, such that the new marking is $M'= M - \preset(t) + \postset(t)$. This is denoted as $M|t\rangle M'$. 
This notation is extended to sequences of enabled transitions $M|t_1\ldots t_n \rangle M'$.
A marking $M$ is \emph{reachable} if there exists a sequence of successively enabled transitions $(t_k)_{k \in \{1,\dots,n\}}$ and $ \mathit{In}|t_1 \ldots t_n\rangle M$.  This sequence can be empty.  The set of all reachable markings of a net $N$ is denoted as $\reachablemarkings(N)$.
A Petri net $N$ is called \emph{safe}, if for all reachable markings $M(p) \leq 1$ holds for all $p\in \mathcal{P}$. Then, $M$ is a subset of $P$. All Petri nets considered in this paper are safe.

A \emph{node} $x$ is a place or a transition $x \in \places \cup \transitions$. 
The binary \emph{flow relation} $\mathcal{F}$ on nodes is defined as follows:
$x\, \mathcal{F}\, y$ if $x \in \preset(y)$.
A node $x$ is a \emph{causal predecessor} of $y$, denoted as $x \leq y$, if $x \flowrelation^+ y$.
 Furthermore, $x \leq x$ holds for all $ x \in \mathcal{P} \cup \mathcal{T}$. Two nodes $x,y \in \mathcal{P} \cup \mathcal{T}$  are \emph{causally related}, if $x \leq y$ or $y\leq x$ holds. We say $x$ is a \emph{causal successor} of $y$, if $y \leq x$ holds. 

Two nodes $ x_1,x_2 \in \mathcal{P} \cup \mathcal{T}$ are \emph{in conflict}, denoted $x_1 \# x_2$, if there exist two transitions $t_1, t_2 \in \mathcal{T}$, $t_1 \neq t_2$ with $pre(t_1) \cap pre(t_2) \neq \emptyset$ and $t_i \leq x_i$, $i = 1,2$. A node $ x \in \mathcal{P} \cup \mathcal{T}$ is in self-conflict if $x\#x$.
Informally speaking, two nodes are in conflict if two transitions exist that share some place in their presets and each node is a causal successor of one of those transitions.
Two nodes $ x,y \in \mathcal{P} \cup \mathcal{T}$ are \emph{concurrent}, denoted $x || y$, if they are neither causally related nor in conflict.

A Petri net $N$ is \emph{finitely preceded}, if for every node $x \in \places \cup \transitions$ the set $\{y \in \places \cup \transitions \mid y \leq x \}$ is finite. That set is the causal history of a node. 
A Petri net $N$ is \emph{acyclic}, if the directed graph $(\places \cup \transitions, \flowrelation)$ is acyclic.	
The following definitions lead to the definition of a branching process. 

An \emph{occurrence net} is a Petri net $N$ with the following properties:
$N$ is acyclic,  finitely preceded, $\forall p \in \mathcal{P}: |\mathit{pre}(p)| \leq 1$,	 no transition $t \in \mathcal{T}$ is in self-conflict, and $\mathit{In} = \{p \in \mathcal{P} \mid \, \mathit{pre}(p) = \emptyset \}$.

A homomorphism from one Petri net to another maps each node to a node such that the preset and postset relations are preserved including the initial marking.
Formally, let $N_1$ and  $N_2$ be two Petri nets. Then a \emph{homomorphism} from $N_1$ to $N_2$ is a mapping $h: \mathcal{P}_1 \cup \mathcal{T}_1 \rightarrow \mathcal{P}_2 \cup \mathcal{T}_2$ with following properties:
$h(\mathcal{P}_1) \subseteq \mathcal{P}_2$ and $h(\mathcal{T}_1) \subseteq \mathcal{T}_2$,
for all transitions $t \in \transitions_1$, $h$ restricted to $\mathit{pre}_1(t)$ is a bijection between $\mathit{pre}_1(t)$ and $\mathit{pre}_2(h(t))$,
for all transitions $t \in \mathcal{T}_1$, $h$ restricted to $\mathit{post}_1(t)$ is a bijection between $\mathit{post}_1(t)$ and $\mathit{post}_2(h(t))$, and
the restriction of $h$ to $\mathit{In}_1$ is a bijection between $\mathit{In}_1$ and $\mathit{In}_2$.
An \emph{isomorphism} is a bijective homomorphism.

A run is represented by a (possibly infinite) firing sequence of transitions. A branching process of a Petri net represents (possibly) multiple runs of the underlying Petri net.

{\textbf{Branching process.}}  A \emph{branching process} of a net $N_0$ is a pair $B = (N, \pi)$, where $N$ is an occurrence net and $\pi$ a homomorphism from $N$ to $N_0$ such that: 
$(*)$ For all $t_1, t_2 \in \mathcal{T}$: if $\mathit{pre}(t_1) = \mathit{pre}(t_2)$ and $\pi(t_1) = \pi(t_2)$, then $t_1 = t_2$.

An example of a Petri net and a branching process is shown in Fig. \ref{fig-example-branching-process1}. 

The notion of the set of all reachable markings of a branching process $B =(N, \pi)$ is extended to $\reachablemarkings(B) = \reachablemarkings(N)$. By convention, a branching Process $B$ has the components $(N_B, \pi_B)$. Throughout this paper, $B_1$ and $B_2$ are branching processes of an underlying net $N_0$.
The property $(*)$ of the definition of a branching process ensures that every run of the Petri net is represented at most once. Informally speaking, a run only consists of concurrent and causally related nodes and a node can be part of multiple runs. Nodes that are in conflict, cannot belong to the same run.

\begin{figure}
	\begin{center}
		\begin{minipage}{0.4\textwidth}
			\begin{tikzpicture}[node distance=1cm,>=stealth',bend angle=45,auto]
				
				\node[tokens= 1, eplace, label={above:$e_1$}] (e1) {};
				\node[transition, left of=e1, label={above:$t_1$}] (t1) {};
				\node[tokens = 1, eplace, right of=e1, label={above:$e_2$}] (e2) {};
				\node[transition, below of=e1, label={left:$t_2$}] (t2) {};
				\node[transition, below of=e2, label={right:$t_3$}] (t3) {};
				\node[eplace, below of= t2, label={left:$e_3$}] (e3) {};
				\node[eplace, below of= t3, label={right:$e_4$}] (e4) {};
				\node[left of =t1, yshift=5mm](l){$N_0:$};
				\draw[->] 
				(t1)  		edge [pre, bend left]               (e1)
				edge [post, bend right]             (e1)
				(t2) 		edge[pre] 							(e1)
				edge[pre]							(e2)
				edge[post] 							(e3)
				(t3) 		edge[pre] 							(e2)
				edge[post]							(e4)
				;
			\end{tikzpicture}
		\end{minipage}
		\begin{minipage}{0.45\textwidth}
			\begin{tikzpicture}[node distance=1cm,>=stealth',bend angle=45,auto]
				
				\node[tokens= 1, eplace, label={above:$e^1_1$}] (e1) {};
				\node[transition, left of=e1, label={above:$t^1_1$}] (t1) {};
				\node[ eplace, left of=t1, label={above:${e^2_1}$}] (e1') {};
				\node[tokens = 1, eplace, right of=e1, label={above:$e^1_2$}] (e2) {};
				\node[transition, below of=e1, label={left:$t^1_2$}] (t2) {};
				\node[transition, below of=e1', label={left:${t^2_2}$}] (t2') {};
				\node[transition, below of=e2, label={right:$t^1_3$}] (t3) {};
				\node[eplace, below of= t2, label={left:$e^1_3$}] (e3) {};
				\node[eplace, below of= t2', label={left:${e^2_3}$}] (e3') {};
				\node[eplace, below of= t3, label={right:$e^1_4$}] (e4) {};
				\node[left of=e1', xshift = 1mm] (b1) {\footnotesize$\bullet$};
				\node[left of=b1, xshift= 8mm] (b2) {\footnotesize$\bullet$};
				\node[left of=b2, xshift= 8mm] (b3) {\footnotesize$\bullet$};
				\node[left of =e1', xshift=-13mm, yshift=5mm](l){$B = (N,\pi):$};
				\draw[->] 
				(t1)  		edge [pre]               (e1)
				edge [post]             (e1')
				(t2) 		edge[pre] 							(e1)
				edge[pre]							(e2)
				edge[post] 							(e3)
				(t2') 		edge[pre] 							(e1')
				edge[pre]							(e2)
				edge[post] 							(e3')
				(t3) 		edge[pre] 							(e2)
				edge[post]							(e4)
				(b1)		edge[pre]				(e1')
				;
			\end{tikzpicture}
		\end{minipage}
	\end{center}
	\caption{A Petri net $N_0$ on the left and a branching process $B = (N, \pi)$ of $N_0$ on the right. Places are shown as circles, transitions as boxes and the preset and postset relations as arrows. The initial marking $\{e_1, e_2\}$ is represented by the black dots, the \emph{tokens}.  The homomorphism $\pi$ from $N$ to $N_0$ is given as  $\pi(e^i_j) = e_j$ and $\pi(t^i_j) = t_j$. The transitions $t^1_1$ and $t^1_2$ are in conflict, i.e., $t^1_1 \# t^1_2$, and also $t^2_2 \# t^1_2$, $t^1_3 \# t^1_2$, $t^2_2 \# t^1_3$.}
	\label{fig-example-branching-process1}
\end{figure}
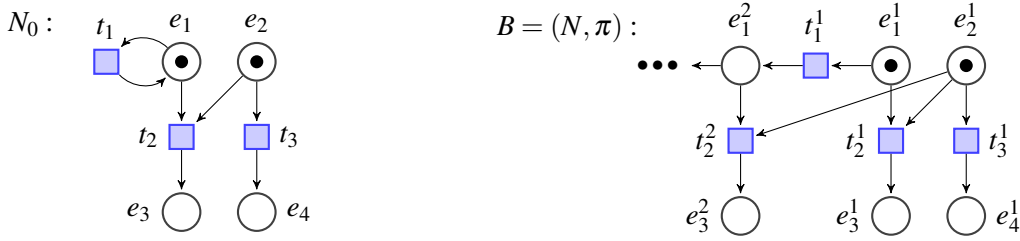
{\textbf{Homomorphism on branching processes.}} 
 A \emph{homomorphism} from $B_1$ to $B_2$ is a homomorphism $h$ from $N_1$ to $N_2$ such that $\pi_2 \circ h = \pi_1$. The branching processes $B_1$ and $B_2$ are $\emph{isomorphic}$ if there exists an isomorphism from $B_1$ to $B_2$ which is denoted as $B_1 \cong B_2$.

A natural partial order on branching processes is defined in the following.

{\textbf{Subprocess relation of branching processes.}} 
$B_1$ approximates $B_2$, denoted by $B_1 \leq B_2$, if there exists an injective homomorphism from $B_1$ to $B_2$.	


Now we define the unfolding of a  net as the maximal branching process that contains all (possibly infinite) runs of a net. Loosely speaking, the unfolding of a  net is an acyclic  net with the same behaviour as the original  net, but where each place and transition has a unique causal history.

{\textbf{Unfolding.}} 
The \emph{unfolding} $\unf(N_0)$ of a Petri net $N_0$ is the maximal branching process with respect to the subprocess relation $\leq$ of branching processes. This definition is unique up to isomorphism.
We refer to the components of the unfolding as $\transitions_{\unf(N_0)}$, $\places_{\unf(N_0)}$, $\preset_{\unf(N_0)}$, $\postset_{\unf(N_0)}$, and $\initmarking_{\unf(N_0)}$.

In Fig. \ref{fig-example-branching-process1}, the branching process $B$ is the unfolding of the Petri net $N_0$ assuming that the dots to the left of the place ${e^2_1}$ indicate that the branching process continues infinitely in the same way.

 We continue with the definition of a Petri game. A Petri game is played on  a finite and safe Petri net. Tokens may transit from a system place to an environment place and vice versa.
 
\textbf{Petri game.}
	A \emph{Petri game} on an underlying finite and safe Petri net $N_0$ is a tuple
	$G = (\mathcal{P}^S_0, \mathcal{P}^E_0 , \mathcal{T}_0,$ $ \preset_0, \postset_0, \mathit{In}_0, \mathcal{B})$, 
	where the set $\mathcal{P}_0$  of places of $N_0$ is partitioned into disjoint sets of \emph{system places} $\mathcal{P}^S_0$,  and 
	\emph{environment places} $\mathcal{P}^E_0$, and where $\mathcal{B} \in 2^{\mathcal{P}_0}$ is the set of \emph{bad markings}.

A winning strategy of a Petri game is a branching process of the underlying Petri net of the game. A strategy must satisfy 4 properties that are reasonable properties of implementations of processes. Beside a \emph{safety} property, each process must act deterministically, \emph{determinism}, and at least one process must enable a transition if possible, \emph{deadlock avoiding}.  Loosely speaking, the \emph{justified refusal} property forces the system players to commit to transitions that are always allowed on their current place.
Strategies for Petri games are obtained by cutting out parts of the unfolding. 

\textbf{Winning strategy}
	A \emph{winning strategy} $\sigma$ of a Petri-game $G = (\mathcal{P}^S_0, \mathcal{P}^E_0 , \mathcal{T}_0, \preset_0, \postset_0, \mathit{In_0}, \mathcal{B})$ with underlying Petri-net $N_0$ is a branching process $\sigma = (N, \pi)$ of $N_0$ satisfying the following properties:
	\begin{enumerate}
		\item \textbf{Justified refusal}:
		Let $C \subseteq \mathcal{P}$ be a set of pairwise concurrent places and $t \in \mathcal{T}_0$ a transition with $\pi(C) = \mathit{pre}_0(t)$. 
		If no $t' \in \mathcal{T}$ with $\pi(t') = t$ and $\mathit{pre}(t') = C$ exists, then there exists a place $p \in C$ with $\pi(p) \in \mathcal{P}^S_0$, such that $t \notin \pi(\mathit{post}(p))$.
		\item \textbf{Safety}: For all reachable markings $M \in \reachablemarkings(N)$ it holds that $\pi(M) \notin \mathcal{B}$.
		\item \textbf{Determinism}: For all $p \in \mathcal{P}$ with $\pi(p) \in \mathcal{P}^S_0$ and for all reachable markings $M$ in $N$ with $p \in M$ there exists at most one transition $t \in \mathit{post}(p)$, which is enabled in $M$. 
		\item \textbf{Deadlock avoiding}: For all reachable markings $M$ in $N$ there exists an enabled transition, if a transition is enabled in $\pi(M)$ in the underlying Petri-net $N_0$. 
	\end{enumerate}

We refer to a token on a system place as a system player, and a token on an environment place as an environment player.
The \emph{justified refusal} property ensures that a system player allows all instances of an outgoing transition or no instance at all.
The global \emph{safety} property ensures that no bad markings are reachable.
The \emph{determinism} property ensures that for each system place at most one transition is enabled in every reachable marking.
The \emph{deadlock avoiding} property ensures that the system allows at least one transition in every reachable marking if an enabled transition exists in that marking, e.g. the system players cannot add deadlocks to the existing deadlocks in the Petri net by forbidding transitions.

The \emph{unfolding} $\unf(G)$ of a Petri game $G$ is like the unfolding $\unf(N_0) =(N,\pi)$ of the underlying Petri net $N_0$ of $G$,
additionally keeping the distinction between system and environment: a place $p$ in $N$ is a system place if
$\pi(p) \in \mathcal{P}^S_0$ and an environment place if $\pi(p) \in \mathcal{P}^E_0$.
A winning strategy $\sigma_G$ of $G$ can be seen as a subprocess of $\unf(G)$.

Fig. \ref{fig-example-introduction} shows a Petri game with 4 players modelling the control of a room with two doors that must not be opened at the same time so that the two places of the bad marking $\{{O^1, O^2}\}$ are not reached simultaneously. After receiving a request to open door via transition $r^1$ the first system player on place $S^{12}$ can decide to proceed with communicating (transition $t$) or without communicating (transition $n^1$) to the second system player, then on place $S^{22}$. On place $S^{13}$ the system player chooses between opening the door ($o^1$) or denying the request ($d^1$) for the environment player waiting on $W^1$. From place $S^{14}$ the system player can close the door it has opened before ($c^1$). The second system player has the same options after receiving a request via $r^2$. After that, if both players open their doors the bad marking $\{{O^1, O^2}\}$ is reached. At least one system player has to deny the request to win the game.

\begin{figure}[ht]
	\begin{center}
		\scalebox{0.9}{
			\begin{tikzpicture}[node distance=1.1cm,>=stealth',bend angle=45,auto]
				
				\node[eplace, label={above:$P^1$}] (openr) {};
				\node[transition, left of=openr, label={above:$p^1$}] (topenr) {};
				\node[tokens=1, eplace, left of=topenr,  label={above:$C^1$}] (eclosed) {};
				
				\node[transition, right of=openr, , label={above:$r^1$}] (tcom){};
				\node[splace, tokens=1, right of=tcom, label={above:$S^{11}$}] (sidle) {};
				\node[splace, below of=sidle,label={above right:$S^{12}$}] (sr) {};
				
				\node[eplace, below of=openr, label={above:$W^1$}](waitr) {};
				\node[transition,below of=sr, label={left:$n^1$}](tnosyscom) {};
				\node[splace, below of=tnosyscom,label={right:$S^{13}$}](saction) {};
				\node[transition, below of=waitr, xshift=-5mm,label={left:$d^1$} ] (tdeny) {};
				\node[transition, below of=tdeny,xshift=5mm, label={left:$o^1$} ] (topen) {};
				\node[eplace, fill=red, left of=topen, xshift=-11mm,yshift=-11mm, label={left:$O^1$} ] (eopen) {};
				\node[transition, below of=topen, label={above:$c^1$} ] (tclose) {};
				\node[splace,right of=tclose,xshift=11mm, label={right:$S^{14}$}] (sclose) {};
				
				\node[transition, right of=sr, label={above:$t$}] (tsyscom) {};
				
				\draw[->] 
				(topenr)  		edge [pre]                            (eclosed)
				edge [post]                           (openr)
				(tcom)  		edge [pre]                            (sidle)
				edge [pre]							  (openr)
				edge [post, bend left]                           (waitr)
				edge [post]                           (sr)
				(tnosyscom)		edge [pre]							(sr)
				edge [post]							(saction)
				(tdeny)			edge [pre]							(waitr)
				edge [pre]							(saction)
				edge [post]							(sidle)
				edge [post]							(eclosed)
				(topen)			edge [pre]							(waitr)
				edge [pre]							(saction)
				edge [post]							(sclose)
				edge [post]							(eopen)
				(tclose)		edge [pre]							(eopen)
				edge [pre]							(sclose)
				edge [post,bend left=50]				(eclosed)
				edge [post]							(sidle)
				;
				
				\node[splace, tokens=1, right of=sidle, xshift=11mm, label={above:$S^{21}$}] (sidle2) {};
				\node[splace, below of=sidle2, label={above left:$S^{22}$}] (sr2) {};
				\node[transition, right of=sidle2, label={above:$r^2$}] (tcom2){};
				\node[eplace, right of=tcom2, label={above:$P^2$}] (openr2) {};
				\node[transition, right of=openr2, label={above:$p^2$}] (topenr2) {};
				\node[tokens=1, eplace, right of=topenr2,  label={above:$C^2$}] (eclosed2) {};


				\node[eplace, below of=openr2, label={above:$W^2$}](waitr2) {};
				\node[transition,below of=sr2, label={right:$n^2$}](tnosyscom2) {};
				\node[splace, below of=tnosyscom2, label={left:$S^{23}$}](saction2) {};
				\node[transition, below of=waitr2, xshift=5mm,label={right:$d^2$} ] (tdeny2) {};
				\node[transition, below of=tdeny2,xshift=-5mm, label={right:$o^2$} ] (topen2) {};
				\node[eplace, fill=red, right of=topen2, xshift=11mm,yshift=-11mm, label={right:$O^2$} ] (eopen2) {};
				\node[transition, below of=topen2, label={above:$c^2$} ] (tclose2) {};
				\node[splace,left of=tclose2,xshift=-11mm, label={left:$S^{24}$}] (sclose2) {};
				
				\node[left of=eclosed, yshift=5mm](l){$G$:};
				
				\draw[->] 
				(topenr2)  		edge [pre]                            (eclosed2)
				edge [post]                           (openr2)
				(tcom2)  		edge [pre]                            (sidle2)
				edge [pre]							  (openr2)
				edge [post, bend right=50]                           (waitr2)
				edge [post]                           (sr2)
				(tnosyscom2)		edge [pre]							(sr2)
				edge [post]							(saction2)
				(tdeny2)			edge [pre]							(waitr2)
				edge [pre]							(saction2)
				edge [post]							(sidle2)
				edge [post]							(eclosed2)
				(topen2)			edge [pre]							(waitr2)
				edge [pre]							(saction2)
				edge [post]							(sclose2)
				edge [post]							(eopen2)
				(tclose2)		edge [pre]							(eopen2)
				edge [pre]							(sclose2)
				edge [post,bend right]				(eclosed2)
				edge [post]							(sidle2)
				(tsyscom)		edge [pre]			(sr)
				edge [pre]			(sr2)
				edge [post]			(saction)
				edge [post]			(saction2)
				;
				\draw[->, dashed] 
				;
			\end{tikzpicture}
		}
	\end{center}
	\caption{A Petri game $G$: the grey places belong to the system players and the white places to the environment players. There are two doors and two system players taking requests to open the door. After a request the system players decide whether to communicate and open their door afterwards. The bad marking $\{O^1, O^2\}$ is reached if the system players decide to open both doors simultaneously. }
	
	\label{fig-example-introduction}
	
\end{figure}
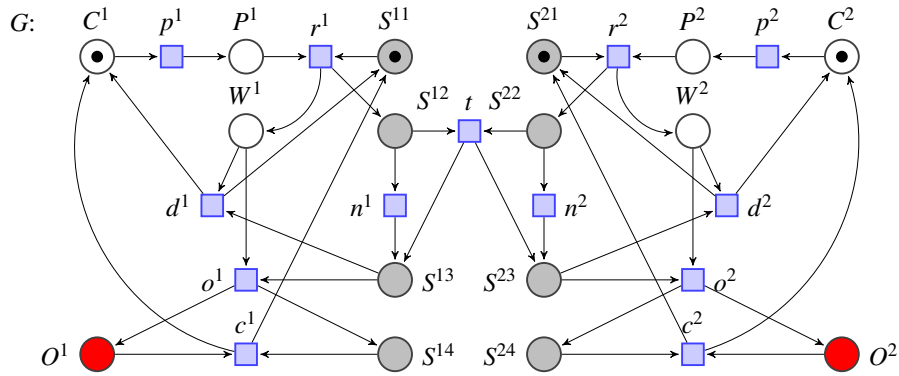

Fig. \ref{fig-example-introduction2} shows an initial part of the unfolding of the Petri game in Fig. \ref{fig-example-introduction}. The parts that are not greyed out are the initial parts of a winning strategy. Here, the system player agree on communicating via $t_1$. The first player decides to open its door via transition $o_2^1$ while the other door remains closed via transition $d_2^2$. After the transition $c_2^1$ the first door is closed again. The following is not shown in Fig. \ref{fig-example-introduction2} anymore: after both players have received another opening request via $r^1_2$ and $r^2_2$, respectively, the winning strategy could continue opening door 2 and keeping door 1 closed since the different causal histories allow different decisions.

 According to the definition of a winning strategy it is also possible that a winning strategy denies all requests to open a door. The example is chosen with foresight for the content in Section \ref{sec-decidability}.

\begin{figure}[ht]
	\begin{center}
		\scalebox{0.83}{
			\begin{tikzpicture}[node distance=1.1cm,>=stealth',bend angle=45,auto]
				
				\node[eplace, label={above:$P^1_1$}] (openr) {};
				\node[transition, left of=openr, label={above:$p^1_1$}] (topenr) {};
				\node[tokens=1, eplace, left of=topenr,  label={above:$C_1^1$}] (eclosed) {};
				
				\node[transition, right of=openr, label={above:$r_1^1$}] (tcom){};
				\node[splace, tokens=1, right of=tcom, label={above:$S^{11}_1$}  ] (sidle) {};
				\node[splace, below of=sidle,  label={above right:$S^{12}_1$}] (sr) {};
				
				\node[eplace, below of=openr, label={above:$W^1_1$}](waitr) {};
				\node[transition,below of=sr, xshift=-11mm, opacity=0.5, label={[opacity=0.5]right:$n^1_1$}](tnosyscom) {};
				\node[splace, left of=tnosyscom,  label={[opacity=0.5]below:$S^{13}_1$}, opacity=0.5](saction) {};
				\node[splace, below of=tnosyscom, xshift=11mm,  label={above right:$S^{13}_2$}](sactionb) {};
				\node[transition, left of=waitr, xshift=-0mm,label={[opacity=0.5]left:$d_1^1$}, opacity=0.5 ] (tdeny) {};
				\node[transition, below of=tdeny,yshift=-11mm, label={[opacity=0.5]right:$o_1^1$}, opacity=0.5 ] (topen) {};
				\node[eplace, fill=red, left of=topen,yshift=-11mm, label={[opacity=0.5]above:$O_1^1$}, opacity=0.5 ] (eopen) {};
				\node[transition, below of=eopen, label={[opacity=0.5]right:$c_1^1$}, opacity=0.5 ] (tclose) {};
				\node[splace,right of=eopen, label={[opacity=0.5]above right:$S^{14}_1$}, opacity=0.5] (sclose) {};
				
				\node[ eplace, below of=tclose,  label={[opacity=0.5]below:$C_3^1$}, opacity=0.5] (eclosedb) {};
				\node[splace, right of=eclosedb,  label={[opacity=0.5]below:$S^{11}_3$}, opacity=0.5] (sidleb) {};
				
				\node[splace, below of=tdeny,  label={[opacity=0.5]below left:$S^{11}_2$}, opacity=0.5] (sidlec) {};
				\node[ eplace, left of=sidlec,  label={[opacity=0.5]left:$C^1_2$}, opacity=0.5] (eclosedc) {};
				
				\node[transition, right of=sclose, xshift=11mm,label={[opacity=0.5]left:$d^1_2$}, opacity=0.5 ] (tdenyb) {};
				\node[transition, right of=tdenyb,xshift=-4mm,label={right:$o^1_2$} ] (topenb) {};
				
				\node[splace, below of=tdenyb, label={},  label={[opacity=0.5]below:$S^{11}_4$}, opacity=0.5] (sidled) {};
				\node[ eplace, left of=sidled,  label={[opacity=0.5]below:$C^1_4$}, opacity=0.5] (eclosedd) {};
				
				\node[eplace, fill=red, below of=topenb, label={right:$O^1_2$} ] (eopenb) {};
				\node[splace,right of=eopenb,  label={above:$S^{14}_2$}] (scloseb) {};
				
				\node[transition, below of=eopenb, label={right:$c^1_2$} ] (tcloseb) {};
				\node[ eplace, below of=tcloseb,  label={left:$C^1_5$}] (eclosede) {};
				\node[splace, right of=eclosede,  label={above:$S^{11}_5$}] (sidlee) {};
				
				\node[transition, right of=sr, label={below:$t_1$}] (tsyscom) {};
				
				\node[left of= sidlee, xshift=5.5mm] (ksidlee) {};
				\node[below of=ksidlee] (lsidlee){};
				\draw[thick, dotted] (ksidlee) -- (lsidlee);
				
				\node[right of= eclosedd, xshift=-5.5mm] (keclosedd) {};
				\node[below of=keclosedd] (leclosedd){};
				\draw[thick, dotted] (keclosedd) -- (leclosedd);
			
				\node[right of= eclosedc, xshift=-5.5mm] (keclosedc) {};
				\node[below left of=keclosedc] (leclosedc){};
				\draw[thick, dotted] (keclosedc) -- (leclosedc); 
				
				\node[left of= sidleb, xshift=5.5mm] (ksidleb) {};
				\node[below of=ksidleb] (lsidleb){};
				\draw[thick, dotted] (ksidleb) -- (lsidleb);
				
				\draw[->, opacity=0.5] 
				
				(tnosyscom)		edge [pre]							(sr)
				edge [post]							(saction)
				(tdeny)			edge [pre]							(waitr)
				edge [pre]							(saction)
				edge [post]							(sidlec)
				edge [post]							(eclosedc)
				(topen)			edge [pre]							(waitr)
				edge [pre]							(saction)
				edge [post]							(sclose)
				edge [post, bend left=5]							(eopen)
				(tclose)		edge [pre]							(eopen)
				edge [pre]							(sclose)
				edge [post]				(eclosedb)
				edge [post]							(sidleb)
				(tdenyb)			edge [pre]							(waitr)
				edge [pre]							(sactionb)
				edge [post]							(sidled)
				edge [post]							(eclosedd)
				
				;
				\draw[->]
				(topenr)  		edge [pre]                            (eclosed)
				edge [post]                           (openr)
				(tcom)  		edge [pre]                            (sidle)
				edge [pre]							  (openr)
				edge [post, bend left]                           (waitr)
				edge [post]                           (sr)
				(topenb)			edge [pre]							(waitr)
				edge [pre]							(sactionb)
				edge [post]							(scloseb)
				edge [post]							(eopenb)
				(tcloseb)		edge [pre]							(eopenb)
				edge [pre]							(scloseb)
				edge [post]				(eclosede)
				edge [post]							(sidlee)
				;
				
				%
				%
				%
				%
				
				\node[splace, tokens=1, right of=sidle, xshift=11mm, label={above:$S^{21}_1$}] (sidle2) {};
				\node[splace, below of=sidle2, label={above left:$S^{22}_1$}] (sr2) {};
				\node[transition, right of=sidle2, label={above:$r_1^2$}] (tcom2){};
				\node[eplace, right of=tcom2, label={above:$P_1^2$}] (openr2) {};
				\node[transition, right of=openr2, label={above:$p^2_1$}] (topenr2) {};
				\node[tokens=1, eplace, right of=topenr2,  label={above:$C_1^2$}] (eclosed2) {};
				
				\node[eplace, below of=openr2, label={above:$W_1^2$}](waitr2) {};
				\node[transition,below of=sr2, xshift=11mm, opacity=0.5, label={left:$n^2_1$}](tnosyscom2) {};
				
				\node[splace, left of=tnosyscom2, xshift=2mm, yshift=-11mm,label={above:$S^{23}_2$}](saction2b) {};
				\node[splace, right of=tnosyscom2, label={[opacity=0.5]below:$S^{23}_1$}, opacity=0.5](saction2) {};
				\node[transition, right of=waitr2, xshift=-0mm,label={[opacity=0.5]right:$d_1^2$}, opacity=0.5 ] (tdeny2) {};
				\node[transition, right of=saction2, yshift=-11mm,label={[opacity=0.5]left:$o_1^2$}, opacity=0.5 ] (topen2) {};
				\node[splace,below of=topen2,label={[opacity=0.5]above left:$S^{24}_1$}, opacity=0.5] (sclose2) {};
				\node[eplace, fill=red, right of=sclose2,  label={[opacity=0.5]above:$O_1^2$}, opacity=0.5 ] (eopen2) {};
				\node[transition, below of=eopen2, label={[opacity=0.5]left:$c_1^2$}, opacity=0.5 ] (tclose2) {};

				\node[ eplace, below of=tclose2,  label={[opacity=0.5]below:$C_3^2$}, opacity=0.5] (eclosed2b) {};
				\node[splace, left of=eclosed2b,label={[opacity=0.5]below:$S^{21}_3$}, opacity=0.5] (sidle2b) {};
				
				\node[splace, below of=tdeny2, label={[opacity=0.5]below right:$S^{21}_2$}, opacity=0.5] (sidle2c) {};
				\node[ eplace, right of=sidle2c,  label={[opacity=0.5]right:$C^2_2$}, opacity=0.5] (eclosed2c) {};
				
				\node[transition, left of=sclose2, xshift=-11mm,label={right:$d_2^2$} ] (tdeny2b) {};
				\node[transition, left of=tdeny2b,xshift=4mm,label={[opacity=0.5] left:$o_2^2$}, opacity=0.5 ] (topen2b) {};
				
				\node[splace, below of=tdeny2b,label={below:$S^{21}_4$}] (sidle2d) {};
				\node[ eplace, right of=sidle2d,  label={below:$C_4^2$}] (eclosed2d) {};
				
				\node[eplace, fill=red, below of=topen2b, label={[opacity=0.5] left:$O^2_2$}, opacity=0.5 ] (eopen2b) {};
				\node[splace,left of=eopen2b, label={[opacity=0.5]above:$s^{24}_2$}, opacity=0.5] (sclosed2b) {};
				
				\node[transition, below of=eopen2b, label={[opacity=0.5]left:$c_2^2$}, opacity=0.5 ] (tclose2b) {};
				\node[ eplace, below of=tclose2b,  label={[opacity=0.5]right:$C_5^2$}, opacity=0.5] (eclosed2e) {};
				\node[splace, left of=eclosed2e, label={[opacity=0.5]above:$S^{21}_5$}, opacity=0.5] (sidle2e) {};
				
				\node[left of=eclosed, yshift=5mm, xshift=-5mm](l){$\sigma_G= (N, \pi)$:};
				
				\node[right of= sidle2e, xshift=-5.5mm] (ksidle2e) {};
				\node[below of=ksidle2e] (lsidle2e){};
				\draw[thick, dotted] (ksidle2e) -- (lsidle2e);
				
				\node[left of= eclosed2d, xshift=5.5mm] (keclosed2d) {};
				\node[below of=keclosed2d] (leclosed2d){};
				\draw[thick, dotted] (keclosed2d) -- (leclosed2d);
				
				\node[left of= eclosed2c, xshift=5.5mm] (keclosed2c) {};
				\node[below right of=keclosed2c] (leclosed2c){};
				\draw[thick, dotted] (keclosed2c) -- (leclosed2c); 
				
				\node[right of= sidle2b, xshift=-5.5mm] (ksidle2b) {};
				\node[below of=ksidle2b] (lsidle2b){};
				\draw[thick, dotted] (ksidle2b) -- (lsidle2b);

				\draw[->, opacity=0.5] 
				
				(tnosyscom2)		edge [pre]							(sr2)
				edge [post]							(saction2)
				(tdeny2)			edge [pre]							(waitr2)
				edge [pre]							(saction2)
				edge [post]							(sidle2c)
				edge [post]							(eclosed2c)
				(topen2)			edge [pre]							(waitr2)
				edge [pre]							(saction2)
				edge [post]							(sclose2)
				edge [post, bend right=15]							(eopen2)
				(tclose2)		edge [pre]							(eopen2)
				edge [pre]							(sclose2)
				edge [post]				(eclosed2b)
				edge [post]							(sidle2b)
				(topen2b)			edge [pre]							(waitr2)
				edge [pre]							(saction2b)
				edge [post]							(sclosed2b)
				edge [post]							(eopen2b)
				(tclose2b)		edge [pre]							(eopen2b)
				edge [pre]							(sclosed2b)
				edge [post]				(eclosed2e)
				edge [post]							(sidle2e)
				;
				\draw [->]
				(topenr2)  		edge [pre]                            (eclosed2)
				edge [post]                           (openr2)
				(tcom2)  		edge [pre]                            (sidle2)
				edge [pre]							  (openr2)
				edge [post, bend right=50]                           (waitr2)
				edge [post]                           (sr2)
				(tsyscom)		edge [pre]			(sr)
				edge [pre]			(sr2)
				edge [post]			(sactionb)
				edge [post]			(saction2b)	
				(tdeny2b)			edge [pre]							(waitr2)
				edge [pre]							(saction2b)
				edge [post]							(sidle2d)
				edge [post]							(eclosed2d)
				;
				\draw[->, dashed] 
				;
			\end{tikzpicture}
		}
	\end{center}
	\caption{An initial part of the unfolding of the Petri game $G$ in Fig.~\ref{fig-example-introduction}. The nodes that are not greyed out form an initial part of a winning strategy $\sigma_G$ where door 1 gets opened and door 2 remains closed after the system players have communicated. The homomorphism $\pi$ is  defined analogously to that in Fig. \ref{fig-example-branching-process1}.}
	\label{fig-example-introduction2}
\end{figure}

%% file: sections/4-decidability-results.tex

In this section, we show that the synthesis problem for Petri games with up to 4 players under a global safety condition is decidable in non-deterministic exponential time (NEXP), and NP-complete in the 2-player case.
In \cite{DBLP:conf/fsttcs/Gimbert17}, a related result is shown that the synthesis problem for 4-process control games with a local termination condition is decidable. The translation \cite{DBLP:conf/concur/BeutnerFH19} of this result to Petri games gives a decidability result for Petri games with 4 players with a local termination condition, without process generation and deletion. There are no complexity bounds for the 4-player case given in \cite{DBLP:conf/fsttcs/Gimbert17}, and the translation to Petri games already generates an exponential blow up in the number of nodes \cite{DBLP:conf/concur/BeutnerFH19}.



%


 A Petri game $G$ is called a \emph{K-player Petri game}, $K \in \mathbb{N}$, if and only if for all reachable markings $M \in \reachablemarkings(G)$ it holds that $|M| \leq K$. Two-player Petri games can be seen as a natural generalisation of infinite games on graphs with a safety winning condition. NP-hardness is established by a reduction of the 3-SAT problem to 2-player Petri games.

\begin{lemma}
	\label{lemma-2playernphard}
	There exists a polynomial-time reduction of the 3-SAT problem to the synthesis problem of  two-player Petri games. 	
\end{lemma}
\begin{proof}
	Let $F = (x^1_1 \lor x^1_2 \lor x^1_3) \land \ldots \land (x^n_1 \lor x^n_2 \lor x^n_3)$ be an instance of the 3-SAT problem in conjunctive normal form where all clauses consist of exactly three literals and where $x^i_j$,  $i= 1, \ldots, n$ and $j=1,2,3$,  is a positive or negative (with overline) literal from the set $\{x_1, \bar{x}_1, \ldots x_m, \bar{x}_m \}$.
	\begin{figure}
		\centering
		\scalebox{0.9}{
		\begin{tikzpicture}[node distance=1.12cm,>=stealth',bend angle=45,auto]
			
			\node[splace, tokens=1] (sc1) {};
			\node[transition, right of=sc1, yshift=5mm, label={above:$$}] (tx1) {};
			\node[transition, below of=tx1, label={above:$$}] (tx1n) {};
			
			\node[ splace, right of=tx1,  label={above:$x_1$}] (sx1) {};
			\node[ splace, right of=tx1n,  label={above:$\bar{x}_1$}] (sx1n) {};
			
			\node[transition, right of=sx1, label={above:$$}] (tc2) {};
			\node[transition, right of=sx1n, label={above:$$}] (tc2n) {};
			
			\node[ splace, right of=tc2, yshift=-5mm,  label={above:$$}] (sc2) {};
			
			\node[right of=sc2, yshift=5mm](h1) {};
			\node[below of=h1](h2) {};
			\node[right of=sc2, xshift=3mm](b1){$\bullet$};
			\node[right of=sc2, xshift=6mm](b2){$\bullet$};
			\node[right of=sc2, xshift=9mm](b3){$\bullet$};
			
			\node[right of=b3,yshift=5mm, xshift=-7mm](h3) {};
			\node[below of=h3](h4) {};

			\node[splace, right of=b3, xshift=3mm] (sc3) {};
			\node[transition, right of=sc3, yshift=5mm, label={above:$$}] (tx3) {};
			\node[transition, below of=tx3, label={above:$$}] (tx3n) {};
			
			\node[ splace, right of=tx3,  label={above:$x_n$}] (sx3) {};
			\node[ splace, right of=tx3n,  label={above:$\bar{x}_n$}] (sx3n) {};
			
			\node[transition, right of=sx3, label={above:$$}] (tc3) {};
			\node[transition, right of=sx3n, label={above:$$}] (tc3n) {};
			
			\node[ splace, right of=tc3, yshift=-5mm,  label={above:$$}] (sc4) {};

			\draw[->] 
			(tx1)  		edge [pre]                            (sc1)
			edge [post]                           (sx1)
			(tx1n) edge [pre]					(sc1)
			edge[post]					(sx1n)
			(tc2) edge[pre]						(sx1)
			edge[post]					(sc2)
			(tc2n) edge[pre]					(sx1n)
			edge[post]					(sc2)
			
			(tx3)  		edge [pre]                            (sc3)
			edge [post]                           (sx3)
			(tx3n) edge [pre]					(sc3)
			edge[post]					(sx3n)
			(tc3) edge[pre]						(sx3)
			edge[post]					(sc4)
			(tc3n) edge[pre]					(sx3n)
			edge[post]					(sc4)
			;
			\draw[->]
			(sc2) -- (h1);
			\draw[->]
			(sc2) -- (h2);
			\draw[->]
			(h3) -- (sc3);
			\draw[->]
			(h4) -- (sc3);
			
			\node[splace, tokens=1, below of=sc1, yshift=-18mm](s1){};
			\node[transition,right of=s1 ](t1){};
			\node[splace,right of=t1, label={above:$x^1_2$}](st1) {};
			\node[transition,above of=t1](t2){};
			\node[splace, right of=t2, label={above:$x^1_1$}] (st2){};
			\node[transition,below of=t1](t3){};
			\node[splace, right of=t3, label={above:$x^1_3$}] (st3){};
			
			\node[transition, right of=st1] (t4){};
			\node[transition, right of=st2] (t5){};
			\node[transition, right of=st3] (t6){};
			
			\node[splace, right of=t4](s2){};
			
			\node[right of=s2, yshift=5mm](h1b) {};
			\node[below of=h1b](h2b) {};
			\node[right of=s2, xshift=3mm](b1b){$\bullet$};
			\node[right of=s2, xshift=6mm](b2b){$\bullet$};
			\node[right of=s2, xshift=9mm](b3b){$\bullet$};
			
			\node[right of=b3b,yshift=5mm, xshift=-7mm](h3b) {};
			\node[below of=h3b](h4b) {};
			
			\node[splace, right of=b3b, xshift=3mm] (s3) {};
			
			\node[transition,right of=s3 ](t1b){};
			\node[splace,right of=t1b, label={above:$x^n_2$}](st1b) {};
			\node[transition,above of=t1b](t2b){};
			\node[splace, right of=t2b, label={above:$x^n_1$}] (st2b){};
			\node[transition,below of=t1b](t3b){};
			\node[splace, right of=t3b, label={above:$x^n_3$}] (st3b){};
			
			\node[transition, right of=st1b] (t4b){};
			\node[transition, right of=st2b] (t5b){};
			\node[transition, right of=st3b] (t6b){};
			
			\node[splace, right of=t4b](s2b){};
			
			\draw[->] 
			(t1)  		edge [pre]                            (s1)
			edge [post]                           (st1)
			(t2) edge [pre]					(s1)
			edge[post]					(st2)
			(t3) edge[pre]						(s1)
			edge[post]					(st3)
			(t4) edge[pre]					(st1)
			edge[post]					(s2)
			
			(t5)  		edge [pre]                            (st2)
			edge [post]                           (s2)
			(t6) edge [pre]					(st3)
			edge[post]					(s2)
			
			(t1b)  		edge [pre]                            (s3)
			edge [post]                           (st1b)
			(t2b) edge [pre]					(s3)
			edge[post]					(st2b)
			(t3b) edge[pre]						(s3)
			edge[post]					(st3b)
			(t4b) edge[pre]					(st1b)
			edge[post]					(s2b)
			
			(t5b)  		edge [pre]                            (st2b)
			edge [post]                           (s2b)
			(t6b) edge [pre]					(st3b)
			edge[post]					(s2b)
			;
			\draw[->]
			(s2) -- (h1b);
			\draw[->]
			(s2) -- (h2b);
			\draw[->]
			(h3b) -- (s3);
			\draw[->]
			(h4b) -- (s3);
			
		\end{tikzpicture}}
		\caption{3-SAT problem $F$ as a Petri game. The top system player chooses sequentially the truth value for $x_1, \ldots x_n$. The bottom system player has to choose a literal for each clause according to the truth values chosen by the top player. The bad markings are $\{ \{x_k, x^i_j\} \mid x^i_j = \bar{x}_k \} \cup \{ \{\bar{x}_k, x^i_j\} \mid x^i_j = x_k \}$. So every time the bottom player chooses a literal that is not chosen by the top system player a bad marking is reached, e.g. the top player chooses $x_1$ and the bottom player $x^1_3$ where $x^1_3 = \bar{x}_1$. }
		\label{fig-np-hard}
	\end{figure}	
	Fig. \ref{fig-np-hard} shows the Petri game of $F$. If $F$ is satisfiable, the top system player chooses the transitions according to the boolean assignment that satisfies $F$, for example $\bar{x}_1$ if the truth value of $x_1$ is false under the boolean assignment. The bottom system player chooses the literal that is true under the boolean assignment in each clause, for example if $x^1_1 = \bar{x}_1$, she may choose $x^1_1$ since the top system player chooses it. Both system players do not know how when the transition of the other player are fired such that it has to be correct for all possible orders of execution, for example the top player could have chosen the truth value for $x_n$ already before the bottom player chose the literal for the first clause. 
	
	Conversely, the top player's choices yield a boolean assignment that satisfies $F$ if the Petri game shown has a winning strategy since the bottom player chooses one literal of each clause that must be true under the boolean assignment.
\end{proof}
Note that the 2-player case is NP-hard even without an environment player. The matching upper bound is established later, along with the complexity of the 4-player case.

The idea of solving 4-player Petri games is to find game states where a winning strategy can be repeated and still win. A few definitions are necessary to formalise repeating a part of a winning strategy. The following definitions about branching processes are equivalent to those in \cite{DBLP:journals/fmsd/EsparzaRV02}. 

\textbf{Cut.}
The reachable markings in a branching process $B$ are called \emph{cuts}. A cut is a maximal set of pairwise concurrent places \cite{DBLP:journals/corr/abs-1710-05368}.  

\textbf{Future of a cut.}
We define the branching process $\future(B , C)$ of a cut $C \subseteq \places_B$ as follows: $\places_{\future(B,C)} \cup \transitions_{\future(B,C)}   = \{x \in \places_B \cup \transitions_B \mid \forall p \in C: p \leq x \lor p || x \}$, the mappings $preset$ and $\postset$ are the mappings $\preset_B$ and $\postset_B$ restricted to $\transitions_{\future(B,C)}$, and $\initmarking_{\future(B,C)} = C$. Generally, the \emph{future} $\future(B,C)$ is not a branching process of the underlying net $N_0$ of $B$ (i.e. $\pi_B: \transitions_B \cup \places_B \to \transitions_0 \cup \places_0)$ but it is a branching process of the net $(\places_0, \transitions_0, \preset_0, \postset_0, \pi_B(C))$, i.e. $\future(B, C)$ is a branching process of $N_0$ if $\pi_B(C) = \initmarking_0$. (This definition is equal to the definition of ${\Uparrow}\mathit{Configuration}$ in \cite{DBLP:journals/fmsd/EsparzaRV02}.) 

Informally speaking, the definition of the future of a cut coincides with the intuition that it is the branching process that follows after that cut.

\textbf{Last known cut and last known marking.}
	The \emph{last known cut} $\lkc(t)$ of a transition $t$ of a branching process $B$ is defined as $\lkc(t) = \{  p \in \places_B \mid p \nless t \land \forall t' \in \preset_B(p): t' \leq t \}$. Informally speaking, this cut is reached by firing all transitions that are causal predecessors of $t$. The \emph{last known marking} $\lkm(t)$ is defined as $\pi_B(\lkc(t))$, which is the marking reached in the underlying Petri net. (The $\lkc$ is the cut of a \emph{local configuration} \cite{DBLP:journals/fmsd/EsparzaRV02})
	
A \emph{cut} and \emph{glue} operation formalises the process of copying the future of one cut to another cut of a strategy. Later, the actual requirements for when to copy are defined.  In the following, $B$ and $B'$ are branching processes of (possibly different) nets.

\textbf{Cut.} 
	$B - B' = (\places_B \setminus (\places_{B'} \setminus \initmarking_{B'}), \transitions_B \setminus \transitions_{B'}, \preset_B\restriction_{\transitions_{B-B'}},  \postset_B\restriction_{\transitions_{B-B'}}, \initmarking_B)$
	
	If $B'$ is the future of a cut of $B$, the cut operation removes $B'$ from $B$ leaving only the initial marking of $B'$ in $B$.

\textbf{Glue.}
$B + B' = (\places_B \cup \places_{B'}, \transitions_B \cup \transitions_{B'}, \preset_{B+B'},  \postset_{B+B'}, \initmarking_B)$, where $\preset_{B+B'}(t) = \begin{cases}
	\preset_B(t) &  t \in \transitions_B \\
	\preset_{B'}(t) &  t \in \transitions_{B'}
\end{cases}$, and analogously $\postset_{B+B'}$.

If the initial marking of $B'$ is a cut in $B$, $B'$ gets glued to that cut.
Generally, $B -B'$ and $B +B'$ are not branching processes. However, if there are two cuts $C_1, C_2 \subseteq \places_B$ with $\pi_B(C_1) = \pi_B(C_2)$ cutting out the future of $C_2$ and glueing an isomorphic copy (this is just a necessary renaming) of the future of $C_1$ to $C_2$ yields a branching process. 

\textbf{Cut and glued branching process.} 
	Let ${\future(B, C_1)}'$ be an isomorphic copy of $\future(B,C_1)$, ${\future(B, C_1)}'$ $ \cong \future(B,C_1)$, such that $\initmarking_{{\future(B, C_1)}'} = C_2$ and  $(\places_B \cup \transitions_B) \cap (\places_{{\future(B, C_1)}'} \cup \transitions_{{\future(B, C_1)}'}) = C_2$. The cut and glued branching process $B_{C_1 \rightarrow C_2}$ is the branching process $B_{C_1 \rightarrow C_2} = B - \future(B, C_2) + {\future(B, C_1)}'$. This definition is unique up to isomorphism.

\begin{definition}[Imitable cuts]
	\label{def-imitable-cuts} Let $\sigma$ be a winning strategy of a Petri game $G$.
	A cut $C_2$ is imitable by a cut $C_1$ if $\sigma_{C_1 \rightarrow C_2}$ is a winning strategy.
\end{definition}


Now we define the actual cuts that are imitable. The first kind of these cuts are cuts where a subset of players take transitions without communicating to the remaining players until all players of this subset synchronise at a joint transition for the second time.
In the following, we fix $\sigma$ as a winning strategy of a Petri game $G$.

\begin{definition}[Partial repetition cuts]	
 \label{def-oneplayerimitablecuts}
		Let $t_1, t_2 \in \transitions_{\sigma}$. The cuts $\lkc(t_1)$ and $\lkc(t_2)$ are \emph{partial repetition cuts}, denoted $\partialrepetition(t_1, t_2)$, if 		
	$\lkm(t_1) = \lkm(t_2) \land \lkc(t_1) \setminus \postset(t_1) = \lkc(t_2) \setminus \postset(t_2) $. 
	The partial repetitions cuts $\partialrepetition(t_1,t_2)$ are called a \emph{loop}, denoted $\loopcuts(t_1, t_2)$ if $t_1 < t_2$. 
\end{definition}

\begin{lemma}[Partial repetition cuts are imitable]	
	\label{lemma-oneplayerimitablecuts}
		For transitions $t_1, t_2 \in \transitions_{\sigma}$, $\partialrepetition(t_1, t_2)$ implies that $\lkc(t_2)$ is imitable by $\lkc(t_1)$ and vice versa.
\end{lemma}
\begin{proof}[Proof by contradiction]	
	We show that $\sigma_{\lkc(t_1) \rightarrow \lkc(t_2)} = \sigma - \future(\sigma, \lkc(t_2)) + {\future(B, \lkc(t_1))}'$ is a winning strategy. Suppose that there exists a cut $C \subseteq \places_{\sigma_{\lkc(t_1)\rightarrow \lkc(t_2)}}$ such that one of the properties of the winning strategy is violated.
	
	We distinguish two cases: The first case is that there exists a place $p \in C$ such that  $t_2 \leq p$. Then, $C \subseteq \places_{\future(\sigma_{\lkc(t_1) \rightarrow \lkc(t_2)}, \lkc(t_2)), }$ holds. Since $\future(\sigma_{\lkc(t_1) \rightarrow \lkc(t_2)} , \lkc(t_2)) = \future(\sigma, \lkc(t_1))' \cong \future(\sigma, \lkc(t_1))$ it follows directly that $\sigma$ is not a winning strategy if $\sigma_{\lkc(t_1) \rightarrow \lkc(t_2)}$  is not winning. 
	
	The second case is that there exists no place $p \in C$ such that $t_2 \leq p$, i.e. $\forall p \in C: p \leq t_2 \lor p \# t_2 \lor p || t_2$ holds. Since the last known cuts of $t_1$ and $t_2$ are equal except for $\postset(t_1)$ and $\postset(t_2)$ the branching processes $\future(\sigma, \lkc(t_2))$ and $\future(\sigma , \lkc(t_1))$ are isomorphic up to the nodes that are causal successors of $t_1$ and $t_2$ respectively. This means that all transitions and places that are not causal successors of $t_2$ are only renamed by constructing $\sigma_{\lkc(t_1) \rightarrow \lkc(t_2)}$. Since there is no place in $C$ that is a causal successor of $t_2$ there is no transition enabled in $C$ that could have been added or removed. It follows that $\sigma$ is not a winning strategy, if $\sigma_{\lkc(t_1) \rightarrow \lkc(t_2)} $ is not a winning strategy. 
\end{proof}

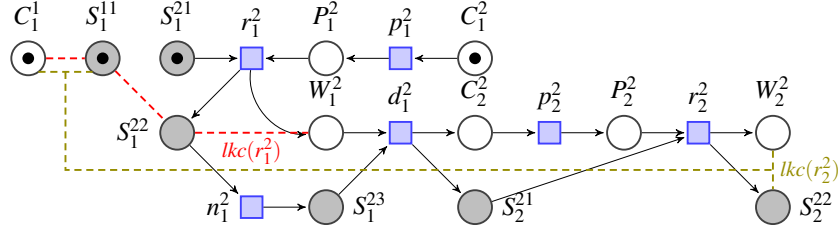
\begin{figure}
	\centering
\scalebox{0.9}{
	\begin{tikzpicture}[node distance=1.1cm,>=stealth',bend angle=45,auto]
		\node[splace, tokens=1, label={above:$S^{21}_1$}] (sidle2) {};
\node[splace, below of=sidle2, label={ left:$S^{22}_1$}] (sr2) {};
\node[transition, right of=sidle2, label={above:$r_1^2$}] (tcom2){};
\node[eplace, right of=tcom2, label={above:$P_1^2$}] (openr2) {};
\node[transition, right of=openr2, label={above:$p^2_1$}] (topenr2) {};
\node[tokens=1, eplace, right of=topenr2,  label={above:$C_1^2$}] (eclosed2) {};

\node[eplace, below of=openr2, label={above:$W_1^2$}](waitr2) {};
\node[transition,below of=sr2, xshift=11mm, opacity=1, label={left:$n^2_1$}](tnosyscom2) {};

\node[splace, right of=tnosyscom2, label={[opacity=1]right:$S^{23}_1$}, opacity=1](saction2) {};
\node[transition, right of=waitr2, xshift=-0mm,label={[opacity=1]above:$d_1^2$}, opacity=1 ] (tdeny2) {};

\node[ eplace, right of=tdeny2,  label={[opacity=1]above:$C^2_2$}, opacity=1] (eclosed2c) {};
\node[splace, below of=eclosed2c, label={[opacity=1]right:$S^{21}_2$}, opacity=1] (sidle2c) {};

\node[transition, right of = eclosed2c, label={above:$p^2_2$}] (t1) {};
\node[eplace, right of=t1, label={above:$P_2^2$}] (e1) {};
\node[transition, right of=e1, label={above:$r_2^2$}] (t2){};

\node[eplace, right of=t2, label={above:$W_2^2$}](e2) {};
\node[splace, below of=e2, label={ right:$S^{22}_2$}] (s1) {};

\node[splace, left of=sidle2, tokens=1, label={above:$S^{11}_1$}](s2) {};
\node[eplace, left of=s2, tokens=1, label={above:$C^{1}_1$}](e3) {};

\node[right of= sr2, yshift=-2.75mm](l1) {\footnotesize \color{red}$\lkc(r^2_1)$};
\node[right of= e2, yshift=-5.5mm, xshift=-5.5mm](l1) {\footnotesize \color{olive}$\lkc(r^2_2)$};
\draw[->, opacity=1] 
(t1) edge[pre] (eclosed2c)
	edge[post] (e1)
(t2) edge[pre] (e1)
edge[pre] (sidle2c)
edge[post] (e2)
edge[post] (s1)
(tnosyscom2)		edge [pre]							(sr2)
edge [post]							(saction2)
(tdeny2)			edge [pre]							(waitr2)
edge [pre]							(saction2)
edge [post]							(sidle2c)
edge [post]							(eclosed2c)

;
\draw [->]
(topenr2)  		edge [pre]                            (eclosed2)
edge [post]                           (openr2)
(tcom2)  		edge [pre]                            (sidle2)
edge [pre]							  (openr2)
edge [post, bend right=50]                           (waitr2)
edge [post]                           (sr2)
;
\draw[->, dashed] 
;

\draw[densely dashed, red, thick]
	(e3) -- (s2) -- (sr2) -- (waitr2);
\draw[densely dashed, olive, thick]
	(e2) -- (s1);
\draw[densely dashed, olive, thick]
(-2.05,-0.2) -- (-1.25,-0.2);

\draw[densely dashed, olive, thick]
(-1.65,-0.2) -- (-1.65,-1.65) -- (8.8, -1.65);
	
\end{tikzpicture}}
\caption{Example of partial repetition cuts. Assume that this is a part of a winning strategy where the second system player denies ($d^2_1$) the first request ($r^2_1$) to open the door without communicating with the other system player via transition $t$. Here, the last known cuts of $r^2_1$ and $r^2_2$ are $\lkc(r^2_1) = \{C^1_1, S^{11}_1, S^{22}_1, W^2_1\}$ and $\lkc(r^2_2) =  \{C^1_1, S^{11}_1, S^{22}_2, W^2_2\}$, so $\lkm(r^2_1) = \lkm(r^2_2)$. Thus, a partial repetition cut (that is a loop) occurs  $\partialrepetition_{\mathit{loop}}(r^2_1, r^2_2)$.
}
\label{fig-prc-example}
\end{figure}
In Fig. \ref{fig-prc-example} is an example of partial repetition cuts.
Note that not all of those partial repetition cuts are loops as the players can reach the same places in the underlying Petri net by taking different transitions. Partial repetition cuts are defined regardless of the number of players in a Petri game, i.e, these cuts are imitable in any Petri game.

\textbf{Maximally repeated strategy.}
 Since the nodes of a branching process are countable we can construct by induction over the partial repetition cuts of a winning strategy $\sigma$ a winning strategy $\sigma^{\mathit{pr}}$ where $\partialrepetition(t_1, t_2)$ implies that $\future(\sigma^{\mathit{pr}}, t_1) \cong \future(\sigma^{\mathit{pr}}, t_2)$. 
  $\sigma^{\mathit{pr}}$ denotes a \emph{maximally repeated strategy} of $\sigma$.

The partial repetition cuts have a useful implication for the case with up to~4 players: if two players do not take any joint transition with one of the two remaining players they will get to a partial repetition cut eventually. 
The idea of the following definition of a synchronisation segment is based on this implication.
Informally speaking, a synchronisation segment describes the part starting from the last known cut of a transition $t$ until all other players are causal successors of this transition or the strategy can be repeated due to the partial repetition cuts. Here, the idea of when to repeat a strategy is that if all players play identically starting from the last known cut of transition $t$ until they get to know of transition $t$ the strategy can be repeated after transition $t$.

\begin{definition}[Synchronisation segment]
The \emph{synchronisation segment} $\synchronisationsegment(t)$ of a transition $t \in \transitions_B$ of a branching process $B$ is defined as a smallest branching process (with respect to $\leq$) such that \\
(1.) $\synchronisationsegment(t)\leq \future(B, \lkc(t))$ and \\ (2.) $\forall t' \in \transitions_{\future(B, \lkc(t))}: t' \notin \transitions_{\synchronisationsegment(t)} \Rightarrow  $ (a) ${ (\forall p \in \preset_B(t'): t \leq p}) \lor $ \\ \hspace*{5.72cm} (b) ${(\exists t_1, t_2, t_3 \in \transitions_{\synchronisationsegment(t)}: \loopcuts(t_1, t_2) \land t_3 \leq t_2 \land \partialrepetition(t', t_3))}$ 
\end{definition}

For a transition $t' \in \future(B, \lkc(t))$ that is not in $\transitions_{\synchronisationsegment(t)}$, all places in its preset are causal successors of $t$ (2.a) or there is a transition $t_3$ with $\partialrepetition(t', t_3)$ within a loop (2.b). So, the parts of a winning strategy that get repeated due to partial repetition cuts are included in the synchronisation segment. An example of a synchronisation segment is shown in Fig. \ref{fig-example-src} and Fig. \ref{fig-prc-example-app2}. \emph{Synchronisation equivalent cuts} are those cuts where the synchronisation segments are isomorphic. 

\begin{figure}[t!]
	\begin{center}
		\scalebox{0.9}{
			\begin{tikzpicture}[node distance=1.1cm,>=stealth',bend angle=45,auto]
				
				\node[eplace, label={above:$P^1_1$}] (openr) {};
				\node[transition, left of=openr, label={above:$p^1_1$}] (topenr) {};
				\node[tokens=1, eplace, left of=topenr,  label={above:$C_1^1$}] (eclosed) {};

				\node[transition, right of=openr, label={above:$r_1^1$}] (tcom){};

				\node[splace, tokens=1, right of=tcom, label={above:$S^{11}_1$}  ] (sidle) {};
				\node[splace, below of=sidle,  label={above right:$S^{12}_1$}] (sr) {};
				
				\node[eplace, below of=openr, label={below:$W^1_1$}](waitr) {};

				\node[splace, below of=sr,  label={left:$S^{13}_2$}](sactionb) {};

				\node[transition, below of=sactionb,xshift=-4mm,label={right:$o^1_2$} ] (topenb) {};
				
				\node[eplace,  below of=topenb, label={left:$O^1_2$} ] (eopenb) {};
				\node[splace,right of=eopenb,  label={right:$S^{14}_2$}] (scloseb) {};
				
				
				\node[transition, right of=sr, label={below:\color{blue} $t_1$}] (tsyscom) {};

				\draw[->]
				(topenr)  		edge [pre]                            (eclosed)
				edge [post]                           (openr)
				(tcom)  		edge [pre]                            (sidle)
				edge [pre]							  (openr)
				edge [post, bend left]                           (waitr)
				edge [post]                           (sr)
				(topenb)			edge [pre]							(waitr)
				edge [pre]							(sactionb)
				edge [post]							(scloseb)
				edge [post]							(eopenb)
				;

				\node[splace, tokens=1, right of=sidle, xshift=11mm, label={above:$S^{21}_1$}] (sidle2) {};
				\node[splace, below of=sidle2, label={above left:$S^{22}_1$}] (sr2) {};
				\node[transition, right of=sidle2, label={above:$r_1^2$}] (tcom2){};
				\node[eplace, right of=tcom2, label={above:$P_1^2$}] (openr2) {};
				\node[transition, right of=openr2, label={above:$p^2_1$}] (topenr2) {};
				\node[tokens=1, eplace, right of=topenr2,  label={above:$C_1^2$}] (eclosed2) {};
				
				\node[eplace, below of=openr2, label={below:$W_1^2$}](waitr2) {};

				\node[splace, right of=sactionb, xshift=11mm, label={right:$S^{23}_2$}](saction2b) {};
				
				\node[transition, below of=saction2b, xshift=4mm,label={right:$d_2^2$} ] (tdeny2b) {};
				
				\node[splace, below of=tdeny2b,label={left:$S^{21}_4$}] (sidle2d) {};
				\node[ eplace, right of=sidle2d,  label={right:$C_4^2$}] (eclosed2d) {};
				\node[left of=topenb, xshift=-5.5mm](l) {\color{blue}$\synchronisationsegment(t_1)$} ;
				
				\node[right of=eclosed2d, yshift=11mm, xshift= 16.5mm, align=left](l3){\scriptsize No more transitions \\ \scriptsize are added due to 2.a.};
	
				\draw [->]
				(topenr2)  		edge [pre]                            (eclosed2)
				edge [post]                           (openr2)
				(tcom2)  		edge [pre]                            (sidle2)
				edge [pre]							  (openr2)
				edge [post, bend right=50]                           (waitr2)
				edge [post]                           (sr2)
				(tsyscom)		edge [pre]			(sr)
				edge [pre]			(sr2)
				edge [post]			(sactionb)
				edge [post]			(saction2b)	
				(tdeny2b)			edge [pre]							(waitr2)
				edge [pre]							(saction2b)
				edge [post]							(sidle2d)
				edge [post]							(eclosed2d)
				;
				\draw[dashed, blue] 
				(-0.55,-0.55) -- (0.55,-0.55) -- (1.65, -1.65) -- (4.95, -1.65) --
				(6.05, -0.55) -- (7.15, -0.55) -- (7.15 , -4.95) -- (-0.55, -4.95) --(-0.55, -0.55);
				
			\end{tikzpicture}
		}
	\end{center}
	\caption{A subprocess of $\sigma_G$ of Fig. \ref{fig-example-introduction2} is shown. The nodes inside the dashed, blue box are the nodes of the synchronisation segment $\synchronisationsegment(t_1)$. The initial marking is $\initmarking_{\synchronisationsegment(t_1)} = \{W^1_1, S^{13}_2, S^{23}_2, W^2_1\}$.
	 This segment ends after the transitions $o^1_2$ and $d^2_2$ since $t_1 < O^1_2,  S^{14}_2, S^{21}_4, C^2_4$. So, if $\sigma_G$ repeats to open the first door and keeping the second door closed after taking the transition $t$ the winning strategy could repeat itself due to synchronisation equivalent cuts.
   }
	\label{fig-example-src}
\end{figure}
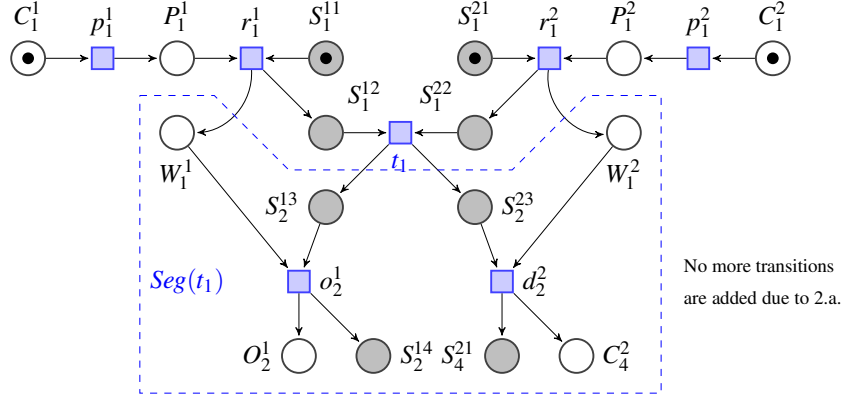

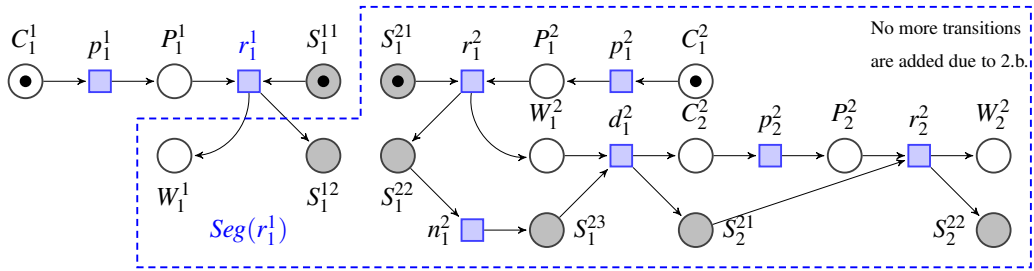
\begin{figure}[ht]
	\centering
	\scalebox{0.9}{
		\begin{tikzpicture}[node distance=1.1cm,>=stealth',bend angle=45,auto]
			\node[splace, tokens=1, label={above:$S^{21}_1$}] (sidle2) {};
			\node[splace, below of=sidle2, label={ below:$S^{22}_1$}] (sr2) {};
			\node[transition, right of=sidle2, label={above:$r_1^2$}] (tcom2){};
			\node[eplace, right of=tcom2, label={above:$P_1^2$}] (openr2) {};
			\node[transition, right of=openr2, label={above:$p^2_1$}] (topenr2) {};
			\node[tokens=1, eplace, right of=topenr2,  label={above:$C_1^2$}] (eclosed2) {};
			
			\node[eplace, below of=openr2, label={above:$W_1^2$}](waitr2) {};
			\node[transition,below of=sr2, xshift=11mm, opacity=1, label={left:$n^2_1$}](tnosyscom2) {};

			\node[splace, right of=tnosyscom2, label={[opacity=1]right:$S^{23}_1$}, opacity=1](saction2) {};
			\node[transition, right of=waitr2, xshift=-0mm,label={[opacity=1]above:$d_1^2$}, opacity=1 ] (tdeny2) {};

			\node[ eplace, right of=tdeny2,  label={[opacity=1]above:$C^2_2$}, opacity=1] (eclosed2c) {};
			\node[splace, below of=eclosed2c, label={[opacity=1]right:$S^{21}_2$}, opacity=1] (sidle2c) {};
			
			\node[transition, right of = eclosed2c, label={above:$p^2_2$}] (t1) {};
			\node[eplace, right of=t1, label={above:$P_2^2$}] (e1) {};
			\node[transition, right of=e1, label={above:$r_2^2$}] (t2){};
			
			\node[eplace, right of=t2, label={above:$W_2^2$}](e2) {};
			\node[splace, below of=e2, label={ left:$S^{22}_2$}] (s1) {};
			
			\node[splace, left of=sidle2, tokens=1, label={above:$S^{11}_1$}](s2) {};
			
			\node[transition, left of=s2, label={above:\color{blue} $r^1_1$}](tn){};
			\node[eplace, left of=tn, label={above:$P^{1}_1$}](e3) {};
			
			\node[transition, left of=e3, label={above:$p^1_1$}](tn2){};
			\node[eplace, left of=tn2, tokens=1, label={above:$C^{1}_1$}](e4) {};
			
			\node[eplace, below of=e3, label={below:$W_1^1$}](w1) {};
			\node[splace, below of=s2, label={ below:$S^{12}_1$}] (sn2) {};
			
			\node[above of=e2,align=left, yshift=5.5mm, xshift=-6.5mm](l3){\scriptsize No more transitions \\ \scriptsize are added due to 2.b.};
			
			\node[below of=w1, xshift=11mm](label) {\color{blue}$\synchronisationsegment(r^1_1)$};
			
			\draw[->, opacity=1] 
			(tn2) edge[pre] (e4)
			edge[post] (e3)
			(tn) edge[pre] (e3)
			edge[pre] (s2)
			edge[post, bend left] (w1)
			edge[post] (sn2)
			(t1) edge[pre] (eclosed2c)
			edge[post] (e1)
			(t2) edge[pre] (e1)
			edge[pre] (sidle2c)
			edge[post] (e2)
			edge[post] (s1)
			(tnosyscom2)		edge [pre]							(sr2)
			edge [post]							(saction2)
			(tdeny2)			edge [pre]							(waitr2)
			edge [pre]							(saction2)
			edge [post]							(sidle2c)
			edge [post]							(eclosed2c)
			
			;
			\draw [->]
			(topenr2)  		edge [pre]                            (eclosed2)
			edge [post]                           (openr2)
			(tcom2)  		edge [pre]                            (sidle2)
			edge [pre]							  (openr2)
			edge [post, bend right=50]                           (waitr2)
			edge [post]                           (sr2)
			;
			\draw[->, dashed] 
			;
			
			\draw[densely dashed, thick, blue]
			(-3.85, -0.55) -- (-0.55, -0.55) -- (-0.55, 1.1) -- (9.35, 1.1) -- (9.35, -2.75)
			-- (-3.85, -2.75) -- (-3.85, -0.55)
			;
			
			%
			%
	\end{tikzpicture}}
	\caption{
		A branching process of the Petri game $G$ of Fig. \ref{fig-example-introduction} is shown. The nodes inside the dashed, blue box are the nodes of the synchronisation segment $\synchronisationsegment(r^1_1)$.
	}
	\label{fig-prc-example-app2}
\end{figure}

\begin{definition}[Synchronisation equivalent cuts]
	Let $\sigma^{\mathit{pr}}$ be a maximally repeated winning strategy of a Petri game $G$ and $t_1, t_2 \in \transitions_{\sigma^{\mathit{pr}}}$. The cuts $\lkc(t_1)$ and $\lkc(t_2)$ are \emph{synchronisation equivalent cuts}, denoted $\synchrequivcuts(t_1, t_2)$, if 	$ \synchronisationsegment(t_1) \cong \synchronisationsegment(t_2)$. $\synchrequivcuts(t_1, t_2)$ is called an \emph{s-loop}, denoted $\synchrequivcuts_{\mathit{loop}}$, if $t_1 < t_2$.
\end{definition}

Now, we show that a winning strategy can be repeated after synchronisation equivalent cuts. 
 
\begin{lemma}[Synchronisation equivalent cuts are imitable]
	\label{lemma-synch-segments}
	For transitions $t_1, t_2 \in \transitions_{\sigma^{\mathit{pr}}}$, $\synchrequivcuts(t_1, t_2)$ implies that $\lkc(t_2)$ is imitable by $\lkc(t_1)$ and vice versa.
\end{lemma}
	
\begin{proof}
	We show that $\sigma^{\mathit{pr}}_{\lkc(t_1) \rightarrow \lkc(t_2)} = \sigma^{\mathit{pr}} - \future(\sigma, \lkc(t_2)) + {\future(\sigma, \lkc(t_1))}'$ is a winning strategy.
	Since $\synchronisationsegment(t_1) \cong \synchronisationsegment(t_2)$ the construction is only a renaming for the nodes in the synchronisation segment and for those that get repeated due to the partial repetition cuts.
	Thus, only nodes that are causal successors of $t_2$ may be removed or added.
	Now, the proof works in the same way as the proof of Lemma \ref{lemma-oneplayerimitablecuts}, that partial repetition cuts are imitable.
\end{proof}
	%
	%
	%
	%
	%
	%
%
We show that 4-player Petri games can be solved in NEXP with the help of Lemma \ref{lemma-oneplayerimitablecuts} and Lemma \ref{lemma-synch-segments}. A winning prefix of sufficient size is guessed from which a winning strategy can be constructed.
A branching process $B$ is a \emph{winning prefix} if there exists a winning strategy $\sigma$ with $B \leq \sigma$.

\begin{theorem}
	\label{theorem-complexity-petri-games}
	The synthesis problem of 2-player Petri games is NP-complete and in NEXP for 3- and 4-player Petri games.
\end{theorem}
\begin{proof}
	Structure of this proof: from the 2-player case over the 3-player case to the 4-player case the size of a winning prefix is determined to guarantee that a winning strategy can be constructed by repeating the futures of imitable cuts. A prefix of appropriate size is guessed and it is checked if it is winning. The prefix does not contain any causal successors of transitions $t_2$ if $\exists t_1: \partialrepetition_\mathit{loop}(t_1, t_2) \lor \synchrequivcuts_{\mathit{loop}}(t_1,t_2)$. Also, it does not contain causal successors of transitions $t_4$, if $\exists t_1, t_2,t_3:  t_3 \leq t_2 \land ((\partialrepetition(t_3, t_4) \land \partialrepetition_\mathit{loop}(t_1,t_2)) \lor (\synchrequivcuts(t_2, t_3) \land \synchrequivcuts_{\mathit{loop}}(t_1,t_2)))$. Let $T$ be the number of transitions of the Petri game. 
	
	
	If a player does not take any joint transition she eventually repeats a place she lays on, i.e. she reaches a partial repetition cut. This occurs after at most $T$ transitions. Otherwise, the two players take a joint transition after at most $T$ transitions. Since there are at most $T$ different joint transitions, a winning prefix from which a winning strategy can be constructed is at most of size $\mathcal{O}(T^2)$. To check if the guessed prefix is winning, we have to check all reachable markings for the winning properties. In a safe Petri net with $T^2$ transitions and at most two tokens there are at most $T^4$ reachable markings. Thus, guessing a winning prefix and checking if it is winning takes time in $\mathcal{O}(|\transitions|^4)$, which shows together with Lemma~\ref{lemma-2playernphard} that the synthesis problem for two-player Petri games is NP-complete. 
	
	In the 3-player case, it holds for the same reasons as in the 2-player case that each pair of two players reaches partial repetition cuts after at most $T^2$ transitions without taking a joint transition with the remaining third player. The remaining third player can take up to $T$ transitions herself before reaching a partial repetition cut or taking a joint transition. Therefore, a synchronisation segment has at most $T^2+T$ transitions resulting in at most $\powerset{(T^2+T)^2}$ non-isomorphic synchronisation segments. Therefore, a synchronisation equivalent cut is always reached after at most $2^{(T^2+T)^2}$ transitions such that the winning prefix can be extended step by step by reaching synchronisation equivalent cuts again and again. Thus, the size of a winning prefix is in $\mathcal{O}(2^{T^4})$. To check if the prefix is winning takes polynomial time in its size and it follows that the 3-player case is decidable in NEXP.
	

	In a 4-player Petri game, it also holds that each pair of two players reaches partial repetition cuts after at most $T^2$ transitions without taking a joint transition with one of the two remaining players. So, there are two pairs of players that can take $T^2$ transitions until they reach partial repetition cuts or one player of each pair communicate with each other. This leaves the other two players  until they reach partial repetition cuts or take a joint transition again after $T^2$ transitions. Thus, there are at most $\powerset{3T^2}$ non-isomorphic synchronisation segments. Therefore, the size of a winning prefix of a 4-player Petri game that needs to be checked if it is winning is in $\mathcal{O}(2^{6T^2})$. So the bound is in NEXP.	
	Note that if more than 2 players take a joint transition the size of the synchronisation segments decreases.
\end{proof}

This construction does not work for the 5-player case. The problem that occurs is tricky to see. Assume there is a group of 3 players and a group of 2 players that take joint transitions repeatedly within their group, only two players participating at a time. The group of 2 players get to partial repetition cuts or one of them takes a joint transition with one player of the other group, while the group of 3 players might not get to partial repetition cuts. After two players, one from each group, have taken a joint transition the problem occurs: the remaining 3 players that did not participate in that joint transition might form a new group of 3 players that do not reach partial repetition cuts. This cannot occur in the 4-player case since there is only a pair of players left.

%% file: sections/3-undecidability-results.tex
In this section, a tiling problem is reduced to the synthesis problem of a 6-player Petri game with a global safety condition. The tiling problem used here is the $\omega$ bipartite colouring problem ($\omega$-BCP). The undecidable $\omega$ Post correspondence problem ($\omega$-PCP) is reducible to the $\omega$-BCP. 
 Later, the undecidability result can be simplified to local safety as a winning condition. The reduction is similar to the work in \cite{DBLP:journals/lmcs/Gimbert22}, where the (normal) Post correspondence problem is reduced to a colouring problem followed by a reduction to the synthesis problem of 6-process control games with local termination as a winning condition. 


Employing the translation of control games into Petri games in
\cite{DBLP:conf/concur/BeutnerFH19},
the 6-process control games yield 6-player Petri games (with
exponentially many additional nodes) for local termination as a winning
condition.
In the obtained Petri game, all players alternate in their roles as
environment and system players.
Instead, we present a direct construction of 6-player Petri games with
at most 3 simultaneous system players.
Later, we show that the number of system players can be even further
reduced to 2.

\begin{definition}[$\omega$ bipartite colourings]
	Let $\colors$ be a finite set of colours. An \emph{$\omega$ bipartite colouring} is a mapping $f: \mathbb{N}^2 \mapsto \colors$. The initial colour is $f(1,1)$. We define three subsets of $\colors^2$ called the patterns induced by $f$:
	\begin{itemize}
		\item the \emph{diagonal patterns} of f are all pairs $\{(f(x,y), f(x+1,y+1)) \mid (x,y) \in \mathbb{N}^2 )\}$	
		\item the \emph{horizontal patterns} of f are all pairs $\{(f(x,y), f(x+1,y)) \mid (x,y) \in \mathbb{N}^2 )\}$
		\item the \emph{vertical patterns} of f are all pairs $\{(f(x,y), f(x,y+1)) \mid (x,y) \in \mathbb{N}^2 )\}$
	\end{itemize}
	We define a \emph{colouring constraint} as a 4-tuple $(\colors_i, \squares, \lowertriangles,\uppertriangles)$, where $\colors_i$ is the set of \emph{initial} colours, $\squares$ a set of diagonal patterns, $\uppertriangles$ a set of vertical patterns and $\lowertriangles$ a set of horizontal patterns. $\squares$, $\uppertriangles$ and $\lowertriangles$ are called \emph{forbidden patterns}. A colouring $f$ \emph{satisfies} a colouring constraint if its initial colour is in $\colors_i$ and no diagonal pattern of $f$ is in $\squares$, no vertical pattern of $f$ is in $\uppertriangles$ and no horizontal pattern of $f$ is in $\lowertriangles$.
\end{definition}
\textbf{$\omega$-BCP.}
	Given a finite set of colours $\colors$ and colouring constraints $(\colors_i, \squares,$ $ \uppertriangles, \lowertriangles)$, decide whether there exists an $\omega$ bipartite colouring that satisfies the colouring constraints. This problem is a variation of the standard tiling problem  and it is also undecidable.
In the following we define a Petri game for which a winning strategy exists if and only if a given $\omega$-BCP has a solution. In this Petri game shown in Fig.~\ref{fig-petri-game-bcp}, there are two identical parts, a top part and a bottom part, each consisting of 3 players. The idea is that the number of \emph{rounds} played in the lower and upper parts refer to the $x$ and $y$ coordinates of a tile, respectively, so that the system players choose a colour for each tile. Each colour choice requires one player from each part of the Petri game, so there can be a maximum of 3 colour choices in one run of the Petri game. If the colours chosen refer to a forbidden pattern a bad marking is reached. As the bad markings have to be defined on the Petri net and not on the unfolding, we define when two system players are in the same round or when a player is one round ahead or one round behind to be able to check for forbidden patterns.
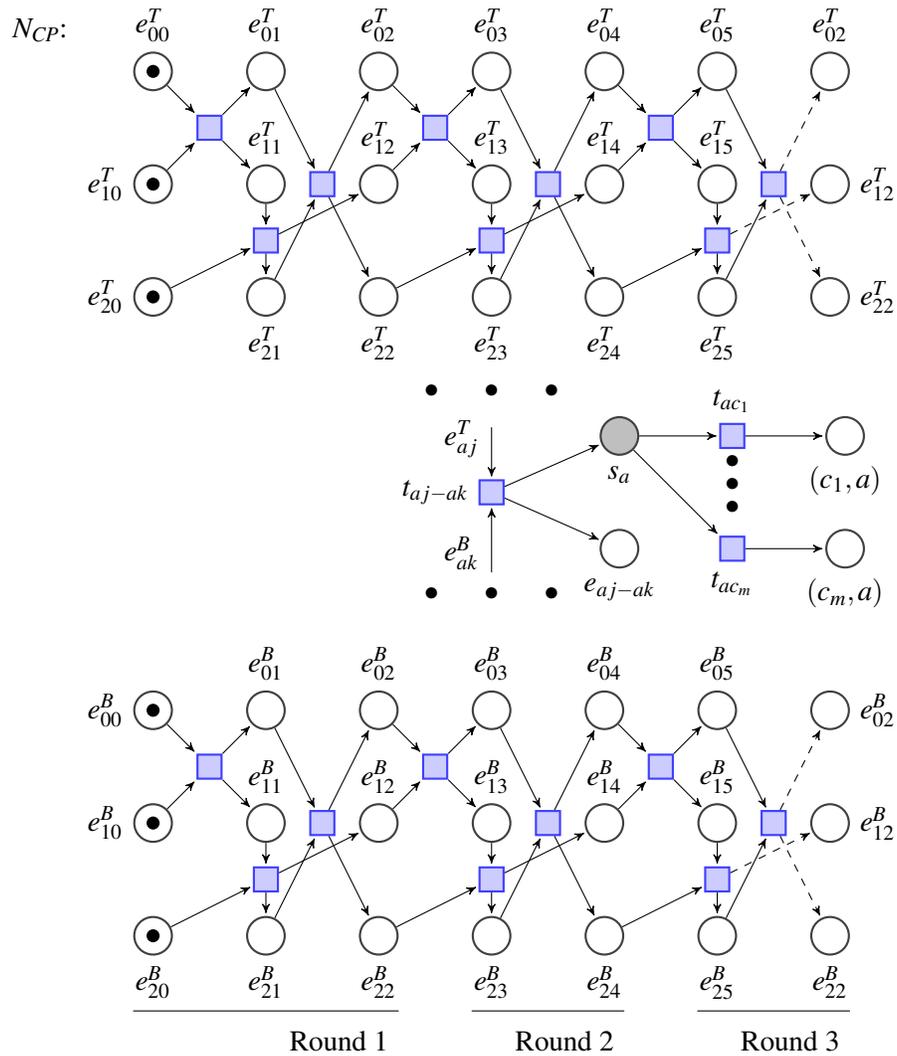
\begin{figure}
	
	\begin{center}
		\begin{tikzpicture}[node distance=1.5cm,>=stealth',bend angle=45,auto]
			
			\node[eplace, tokens=1, label={above:$e^T_{00}$}] (e00) {};
			\node[eplace, right of=e00, label={above:$e^T_{01}$}] (e01) {};
			\node[eplace, right of=e01, label={above:$e^T_{02}$}] (e02) {};
			\node[eplace, right of=e02, label={above:$e^T_{03}$}] (e03) {};
			\node[eplace, right of=e03, label={above:$e^T_{04}$}] (e04) {};
			\node[eplace, right of=e04, label={above:$e^T_{05}$}] (e05) {};
			\node[eplace, right of=e05, label={above:$e^T_{02}$}] (e06) {};
			
			\node[eplace, below of=e00,tokens=1, label={left:$e^T_{10}$}] (e10) {};
			\node[eplace, right of=e10, label={above:$e^T_{11}$}] (e11) {};
			\node[eplace, right of=e11, label={above:$e^T_{12}$}] (e12) {};
			\node[eplace, right of=e12, label={above:$e^T_{13}$}] (e13) {};
			\node[eplace, right of=e13, label={above:$e^T_{14}$}] (e14) {};
			\node[eplace, right of=e14, label={above:$e^T_{15}$}] (e15) {};
			\node[eplace, right of=e15, label={right:$e^T_{12}$}] (e16) {};
			
			\node[eplace, below of=e10, tokens=1, label={left:$e^T_{20}$}] (e20) {};
			\node[eplace, right of=e20, label={below:$e^T_{21}$}] (e21) {};
			\node[eplace, right of=e21, label={below:$e^T_{22}$}] (e22) {};
			\node[eplace, right of=e22, label={below:$e^T_{23}$}] (e23) {};
			\node[eplace, right of=e23, label={below:$e^T_{24}$}] (e24) {};
			\node[eplace, right of=e24, label={below:$e^T_{25}$}] (e25) {};
			\node[eplace, right of=e25, label={right:$e^T_{22}$}] (e26) {};
			
			\node[transition, right of=e00, yshift=-7.5mm, xshift=-7.5mm , label={left:$$}] (t1) {};
			\node[transition, below of=e11 , yshift=7.5mm,label={left:$$}] (t2) {};
			\node[transition, right of=e11 ,xshift=-7.5mm, label={left:$$}] (t3) {};
			
			\node[transition, right of=e02, yshift=-7.5mm, xshift=-7.5mm , label={left:$$}] (t4) {};
			\node[transition, below of=e13 , yshift=7.5mm,label={left:$$}] (t5) {};
			\node[transition, right of=e13 ,xshift=-7.5mm, label={left:$$}] (t6) {};
			
			\node[transition, right of=e04, yshift=-7.5mm, xshift=-7.5mm , label={left:$$}] (t7) {};
			\node[transition, below of=e15 , yshift=7.5mm,label={left:$$}] (t8) {};
			\node[transition, right of=e15 ,xshift=-7.5mm, label={left:$$}] (t9) {};

			\node[eplace, tokens=1, below of=e20, yshift=-40mm, label={left:$e^B_{00}$}] (e00b) {};
			\node[eplace, right of=e00b, label={above:$e^B_{01}$}] (e01b) {};
			\node[eplace, right of=e01b, label={above:$e^B_{02}$}] (e02b) {};
			\node[eplace, right of=e02b, label={above:$e^B_{03}$}] (e03b) {};
			\node[eplace, right of=e03b, label={above:$e^B_{04}$}] (e04b) {};
			\node[eplace, right of=e04b, label={above:$e^B_{05}$}] (e05b) {};
			\node[eplace, right of=e05b, label={right:$e^B_{02}$}] (e06b) {};
			
			\node[eplace, below of=e00b,tokens=1, label={left:$e^B_{10}$}] (e10b) {};
			\node[eplace, right of=e10b, label={above:$e^B_{11}$}] (e11b) {};
			\node[eplace, right of=e11b, label={above:$e^B_{12}$}] (e12b) {};
			\node[eplace, right of=e12b, label={above:$e^B_{13}$}] (e13b) {};
			\node[eplace, right of=e13b, label={above:$e^B_{14}$}] (e14b) {};
			\node[eplace, right of=e14b, label={above:$e^B_{15}$}] (e15b) {};
			\node[eplace, right of=e15b, label={right:$e^B_{12}$}] (e16b) {};
			
			\node[eplace, below of=e10b, tokens=1, label={below:$e^B_{20}$}] (e20b) {};
			\node[eplace, right of=e20b, label={below:$e^B_{21}$}] (e21b) {};
			\node[eplace, right of=e21b, label={below:$e^B_{22}$}] (e22b) {};
			\node[eplace, right of=e22b, label={below:$e^B_{23}$}] (e23b) {};
			\node[eplace, right of=e23b, label={below:$e^B_{24}$}] (e24b) {};
			\node[eplace, right of=e24b, label={below:$e^B_{25}$}] (e25b) {};
			\node[eplace, right of=e25b, label={below:$e^B_{22}$}] (e26b) {};
			
			\node[transition, right of=e00b, yshift=-7.5mm, xshift=-7.5mm , label={left:$$}] (t1b) {};
			\node[transition, below of=e11b , yshift=7.5mm,label={left:$$}] (t2b) {};
			\node[transition, right of=e11b ,xshift=-7.5mm, label={left:$$}] (t3b) {};
			
			\node[transition, right of=e02b, yshift=-7.5mm, xshift=-7.5mm , label={left:$$}] (t4b) {};
			\node[transition, below of=e13b , yshift=7.5mm,label={left:$$}] (t5b) {};
			\node[transition, right of=e13b ,xshift=-7.5mm, label={left:$$}] (t6b) {};
			
			\node[transition, right of=e04b, yshift=-7.5mm, xshift=-7.5mm , label={left:$$}] (t7b) {};
			\node[transition, below of=e15b , yshift=7.5mm,label={left:$$}] (t8b) {};
			\node[transition, right of=e15b ,xshift=-7.5mm, label={left:$$}] (t9b) {};
			
			\node[transition, below of=e23 ,yshift=-11mm, label={left:$t_{aj-ak}$}] (tc) {};
			\node[eplace, xshift=2mm, right of=tc,yshift=-7.5mm, label={below:$e_{aj-ak}$}] (ec) {};
			\node[splace, xshift=2mm, right of=tc,yshift=7.5mm, label={below:$s_a$}] (sc) {};
			\node[transition, right of=sc , label={above:$t_{ac_1}$}] (tc2) {};
			\node[transition, below of=tc2 , label={below:$t_{ac_m}$}] (tc3) {};
			\node[eplace, right of=tc2, label={below:$(c_1,a)$}] (ec2) {};
			\node[eplace, right of=tc3, label={below:$(c_m,a)$}] (ec3) {};
			
			\node[below of=e23, yshift=-1mm,label={$\bullet$}] (dot1) {};
			\node[left of=dot1, xshift=7mm, label={$\bullet$}] (dot2) {};
			\node[left of=dot1, xshift=11mm, yshift=-8mm, label={$e^T_{aj}$}] (dot2) {};
			\node[right of=dot1, xshift=-7mm, label={$\bullet$}] (dot3) {};	
			
			\node[above of=e03b, yshift=-3mm,label={$\bullet$}] (dot4) {};
			\node[left of=dot4, xshift=7mm, label={$\bullet$}] (dot5) {};
			\node[left of=dot4, xshift=11mm,yshift=4mm, label={$e^B_{ak}$}] (dot5b) {};
			\node[right of=dot4, xshift=-7mm, label={$\bullet$}] (dot6) {};	
			
			\node[below of=tc2, yshift=8mm,label={$\bullet$}] (dot7) {};
			\node[below of=tc2, yshift=5mm, label={$\bullet$}] (dot8) {};
			\node[below of=tc2, yshift=2mm, label={$\bullet$}] (dot9) {};	
			
			\node[above of=dot4, yshift=-10mm] (h2) {};
			\node[below of=dot1, yshift=15mm] (h1) {};
			
			\node[left of=dot2, yshift=5mm] (h3) {};
			\node[right of=dot3, yshift=5mm] (h4) {};
			
			\node[below of=e20b, yshift=5mm, xshift=-4mm] (r11) {};
			\node[below of=e22b, yshift=5mm, xshift=4mm,label={below left:Round 1}] (r12) {};
			
			\node[below of=e23b, yshift=5mm, xshift=-4mm] (r21) {};
			\node[below of=e24b, yshift=5mm, xshift=4mm,label={below left:Round 2}] (r22) {};
			
			\node[below of=e25b, yshift=5mm, xshift=-4mm] (r31) {};
			\node[below of=e26b, yshift=5mm, xshift=4mm,label={below left:Round 3}] (r32) {};
			
	\node[left of= e00, yshift=5.5mm](label){$N_\coloringproblem$:};
			
			\draw[->] 
			(t1)  		edge [pre]                            (e00)
			edge [pre]                            (e10)
			edge [post]                           (e01)
			edge [post]                           (e11)
			(t2)  		edge [pre]                            (e20)
			edge [pre]                            (e11)
			edge [post]                           (e21)
			edge [post]                           (e12)
			(t3)  		edge [pre]                            (e01)
			edge [pre]                            (e21)
			edge [post]                           (e02)
			edge [post]                           (e22)
			(t4)  		edge [pre]                            (e02)
			edge [pre]                            (e12)
			edge [post]                           (e03)
			edge [post]                           (e13)
			(t5)  		edge [pre]                            (e13)
			edge [pre]                            (e22)
			edge [post]                           (e14)
			edge [post]                           (e23)
			(t6)  		edge [pre]                            (e03)
			edge [pre]                            (e23)
			edge [post]                           (e04)
			edge [post]                           (e24)
			(t7)  		edge [pre]                            (e04)
			edge [pre]                            (e14)
			edge [post]                           (e05)
			edge [post]                           (e15)
			(t8)  		edge [pre]                            (e24)
			edge [pre]                            (e15)
			edge [post]                           (e25)
			(t9)  		edge [pre]                            (e05)
			edge [pre]                            (e25)	
			
			(t1b)  		edge [pre]                            (e00b)
			edge [pre]                            (e10b)
			edge [post]                           (e01b)
			edge [post]                           (e11b)
			(t2b)  		edge [pre]                            (e20b)
			edge [pre]                            (e11b)
			edge [post]                           (e21b)
			edge [post]                           (e12b)
			(t3b)  		edge [pre]                            (e01b)
			edge [pre]                            (e21b)
			edge [post]                           (e02b)
			edge [post]                           (e22b)
			(t4b)  		edge [pre]                            (e02b)
			edge [pre]                            (e12b)
			edge [post]                           (e03b)
			edge [post]                           (e13b)
			(t5b)  		edge [pre]                            (e13b)
			edge [pre]                            (e22b)
			edge [post]                           (e14b)
			edge [post]                           (e23b)
			(t6b)  		edge [pre]                            (e03b)
			edge [pre]                            (e23b)
			edge [post]                           (e04b)
			edge [post]                           (e24b)
			(t7b)  		edge [pre]                            (e04b)
			edge [pre]                            (e14b)
			edge [post]                           (e05b)
			edge [post]                           (e15b)
			(t8b)  		edge [pre]                            (e24b)
			edge [pre]                            (e15b)
			edge [post]                           (e25b)
			(t9b)  		edge [pre]                            (e05b)
			edge [pre]                            (e25b)
			(tc)		edge [pre]							  (h1)
			edge [pre]							  (h2)
			edge [post]							  (sc)
			edge [post]							  (ec)
			(tc2) 		edge [pre]							  (sc)
			edge [post]							  (ec2)
			(tc3) 		edge [pre]							  (sc)
			edge [post]							  (ec3)
			;
			\draw[dashed]
			(t9)			edge [post]                           (e06)
			edge [post]                           (e26)
			(t8) 			edge [post]                           (e16)
			(t9b)			edge [post]                           (e06b)
			edge [post]                           (e26b)
			(t8b) 			edge [post]                           (e16b)
			;
			\draw[]
			(r11) -- (r12)
			(r21) -- (r22)
			(r31) -- (r32)
			;
		\end{tikzpicture}
	\end{center}
	
	\caption{Petri net $N_\coloringproblem$ of the Petri game $G_\coloringproblem$ of an $\omega$-BCP. We have two structurally identical parts consisting of 3 environment players, a top part and a bottom part. Two of the players synchronise at each transition in these parts. After firing 3 such transitions, each player has synchronised with the other two players in its part. We define such a block of 3 transitions as a \emph{round}. Each part has 3 rounds. After the third round the players return to the places they were in before the second round.
	This means that we have an initial round and then the players alternate between the second and third rounds. As the game progresses, the nodes have exact information in the unfolding about how many rounds they have played. In the middle part, the environment players can take a \emph{check} transition $t_{aj-ak}$, where $a \in \{0,1,2\}$ and $j,k \in \{0, \ldots ,5\}$ and $\preset(t_{aj-ak}) = \{ e^T_{aj}, e^B_{ak} \}$.
	After that, the system player on place $s_a$ has to choose a colour. Based on the number of rounds $x$ and $y$ played in the lower part and upper part, respectively, the system player has to determine the colour $f(x,y)$ of the given $\omega$-BCP. 
 	}
	\label{fig-petri-game-bcp}
\end{figure}

\textbf{Compared rounds of two players.}	
In the Petri net $N_\coloringproblem$ of Fig. \ref{fig-petri-game-bcp}, we compare the rounds of two players $a$ and $b$, both from the same part. We say $a, b \in \{0,1,2\}$ are in the same round if and only if for their places $e^X_{aj}$ and $e^X_{bk}$, $X \in \{T,B\}$, $j,k \in \{0,1,2\}$ or $j,k \in \{3,4\}$ or $j,k \in \{5,2\}$ holds. We say that $a$ is one round ahead of $b$ (or $b$ is one round behind of $a$) if and only if ($j \in \{3\}$ and $k \in \{2\}$) or ($j \in \{5\}$ and $k \in \{4\}$) . Other combinations of indices are not possible.



\begin{definition}[Petri game of an $\omega$ Bipartite colouring Problem]
	\label{def-petri-game-coloring-problem}
	Let $\coloringproblem$ be an $\omega$-BCP with colours $\colors$ and constraints $\colors_i$, $\squares$, $\uppertriangles$ and $\lowertriangles$. The components of the Petri net of the Petri game $G_\coloringproblem$ are the components of $N_\coloringproblem$ from Fig. \ref{fig-petri-game-bcp}. The bad markings of the Petri game $G_\coloringproblem$ are defined as follows:
	
	$\badmarkings = \badmarkings_{\mathit{Same}} \cup \badmarkings_\squares \cup \badmarkings_\lowertriangles \cup \badmarkings_\uppertriangles \cup \badmarkings_\mathit{init}$, where \\
	$\badmarkings_\mathit{Same} = \{ \{ e_{a^Tj-a^Bj'}, e_{b^Tk-b^Bk'},  (c_{u},a), (c_{v},b)\} \mid \\ a^T \text{ and } b^T\text{ are in the same round and } a^B  \text{ and } b^B \text{ are in same the round }, c_u \neq c_v\}$, \\
	$\badmarkings_\mathit{\squares} = \{ \{ e_{a^Tj-a^Bj'}, e_{b^Tk-b^Bk'},  (c_{u},a), (c_{v},b) \mid \\ a^T \text{ and } a^B\text{ are one round ahead of } b^T \text{ and } b^B \text{ respectively } , (c_v, c_u)  \in \squares \}$,\\
	$\badmarkings_\uppertriangles = \{ \{\{ e_{a^Tj-a^Bj'}, e_{b^Tk-b^Bk'},  (c_{u},a), (c_{v},b) \} \mid \\ a^T\text{ is a round ahead of }  b^T \text{ and } a^B \text{ and } b^B\text{ are in the same round }, (c_v, c_u)  \in \uppertriangles \}\}$,\\
	$\badmarkings_\lowertriangles = \{ \{ \{ e_{a^Tj-a^Bj'}, e_{b^Tk-b^Bk'},  (c_{u},a), (c_{v},b) \} \mid \\ a^B\text{ is a round ahead of }  b^B \text{ and } a^T \text{ and } b^B \text{ are in the same round }, (c_v, c_u)  \in \lowertriangles \}\}$,\\
	$\badmarkings_{\mathit{init}} = \{ \{ e_{a^Tj-a^Bj'}, (c_u,a) \} \mid  j,j' \in \{0,1\}, c_u \notin \colors_i \}$.
\end{definition}

In addition to the colouring constraints, the bad markings contain the set $\badmarkings_\mathit{Same}$ to ensure that system players must choose the same colour for two checks whenever the checks occur in the same round for both players.
To define a colouring from a strategy of the Petri game, we define the number of rounds played in the unfolding. 

\textbf{Number of rounds.}
Let $\unf(G_\coloringproblem)$ be the unfolding of $G_\coloringproblem$ of Fig. \ref{fig-petri-game-bcp}. The \emph{number of rounds} played in part $X$, $X \in \{T,B\}$ in a place $p \in \places_{\unf(G_\coloringproblem)}$ is defined as $r^X(p) =  \mathit{max}(\lceil |\{ t \in \transitions_{\unf_\coloringproblem} \mid t \leq p \text{ and } t \text{ is a transition in part } X \}| / 3\rceil, 1)$ 

This means that the number of rounds is incremented after a player has taken two transitions,  starting from round 1. The divisor is 3 because the players also get the knowledge of the synchronisation of the other two players in their part. The number of rounds is not increased after a check transition.

\begin{theorem}[Characterization of $\omega$-BCP as a Petri game]
The $\omega$-BCP is reducible to the synthesis Problem of 6-player Petri games.
\end{theorem}
\begin{proof}
	Let $\coloringproblem$ be a colouring problem and $G_\coloringproblem$ be the Petri game from Def. \ref{def-petri-game-coloring-problem}. We show that there exists a solution for $\coloringproblem$ if and only if there exists a winning strategy for $G_\coloringproblem$.
	
	We start by assuming that there is a solution $f: \mathbb{N} \times \mathbb{N} \mapsto \colors$ for the colouring problem $\coloringproblem$. After the environment player has chosen a check transition the system player must choose a colour in place $s_a$, $a \in \{0,1,2\}$. 
	Since the system player knows the entire causal history, she knows the number of rounds played in the top and the bottom part of the Petri game. Let $r^T(s_a)= y$ denote the number of rounds in the top part and $r^B(s_a)= x$ the number of rounds in the bottom part. 
	Then, the winning strategy for the system players is to choose the colour $f(x,y)$. This way no bad marking is reachable. The bad markings in $\badmarkings_{\mathit{Same}}$ are avoided because $f(x,y)$ is unique. 
	The bad markings in $\badmarkings_\squares, \badmarkings_\uppertriangles$ and $\badmarkings_\lowertriangles$ are avoided because $f$ avoids all forbidden patterns. Finally, the bad markings $\badmarkings_{\mathit{init}}$ are avoided because $f(1,1)$ is an initial colour.

	Now, we assume that there is a winning strategy for the Petri game $G_\coloringproblem$. We define the colouring $f: \mathbb{N} \times \mathbb{N} \mapsto C$ derived from the winning strategy as follows: $f(x,y) = c$ if there exists a system place $s_a \in \places_{\unf(N_\coloringproblem)}$ with $r^B(s_a) = x$ and $r^T(s_a)= y$ and $t_{ac} \in \postset(s_a)$. 
	
	We show that for every colour that a system player has to choose, there are sequences of checks such that all colouring constraints are met. In these sequences of checks, any two consecutive checks can be carried out concurrently while the system players do not know whether the other check is the previous check or the next check, so the winning strategy must play correctly for both checks. 
	Informally speaking, this leads to infinite sequences of dependencies as a check depends on the choice of the colour of the previous check and again that check depends on its previous check and so on. 
	
		\begin{figure}[ht]
		\centering
		\scalebox{0.8}{
			\begin{tabular}{|c|c|c|c|c|c|c|c|c|c|c|c|c|c|c|c|c|c|c|c|c|c|}
				\hline 
				& $e^X_{00}$ & $e^X_{10}$ & $e^X_{20}$ & $e^X_{01}$ & $e^X_{11}$ & $e^X_{21}$ & $e^X_{02}$ & $e^X_{12}$ & $e^X_{22}$ & $e^X_{03}$ & $e^X_{13}$ & $e^X_{23}$  & $e^X_{04}$&$e^X_{14}$ &$e^X_{24}$ \\ 
				\hline 
				$e^X_{00}$ &  & $||$ &$ || $&  &  &  &  &  &  &  &  & & & &\\ 
				\hline 
				$e^X_{10} $& $ ||$ &  & $||$ &  &  &  &  &  &  &  &  & & & &  \\ 
				\hline 
				$e^X_{20}$ &$ ||$ & $|| $ &  &$ ||$  & $||$  &  &  &  &  &  &  &  & & & \\ 
				\hline 
				$e^X_{01}$ &  &  & $||$ &  & $||$ & $||$ &  & $||$  &  &  &  & & & & \\ 
				\hline 
				$e^X_{11}$ &  &  & $||$ & $||$  &  & & &  &  &  &  & & & & \\ 
				\hline 
				$	e^X_{21}$ &  &  &  & $||$ & &  &  & $||$  &  &  &  & & & & \\ 
				\hline 
				$e^X_{02} $&  &  &  &  & &  &  & $||$ & $||$ &  &  & & & &  \\ 
				\hline 
				$e^X_{12}$ &  &  &  & $||$  &  & $||$  & $||$ &  & $||$ &  &  & & & &\\ 
				\hline 
				$e^X_{22}$ &  &  &  &  &  &  & $||$ & $||$ &  & $||$ & $||$ & & & & \\ 
				\hline 
				$e^X_{03}$ &  &  &  &  &  &  &  &  & $||$ &  & $||$ & $||$ & & $||$ &  \\ 
				\hline 
				$e^X_{13}$ &  &  &  &  &  &  &  &  & $||$ & $||$ &  & & & & \\ 
				\hline 
				$e^X_{23}$ &  &  &  &  &  &  &  &  &  & $||$  &  &  & & $||$ & \\ 
				\hline 
				$e^X_{04}$ &  &  &  &  &  &  &  &  &  &  &  & & &$||$  &$||$  \\ 
				\hline 
				$e^X_{14}$ &  &  &  &  &  &  &  &  &  & $||$  &  &$||$  & $||$ & & \\ 
				\hline 
				$e^X_{24}$ &  &  &  &  &  &  &  &  &  &  &  & & $||$ & $||$ & \\ 
				\hline 
				
			\end{tabular}
		}
		\caption{Table of concurrent places of the first two rounds in the unfolding of the Petri game $G_\coloringproblem$. Empty cells denote that the places are causally related. Filled cells denote that the places are concurrent. The pattern continues identically. With the help of this table we can see the sequences of checks to ensure that no bad markings are reached in the Petri game. For example, to check that the initial colour chosen after the transition $t_{00-00}$ and after $t_{22-00}$ is the same, we can perform the following sequence of checks of the indices of the places in the upper part: $00-20-01-12-22$. The indices of the place in the lower part remain $00$ here.
			To carry out two consecutive checks in the same play they need to be concurrent. We can check this in the table by going in a row from $00$ to $20$, then in the column from $20$ to $01$, then again checking in the row and so on. Note that the place of the bottom player can also change in these sequences.  }
		\label{fig-concurrency-table}
		
	\end{figure}	
	
	In Fig. \ref{fig-concurrency-table}, it is crucial to see that there are concurrent places between different rounds such that we can always find at least one suitable check for all diagonal, vertical and horizontal patterns. To check the horizontal pattern of $(f(1,1), f(2,1))$ for example, the checks $t_{00-03}$ and $t_{20-22}$ are fired. The places in the upper part are $e^T_{00}$ and $e^T_{20}$ respectively, which are both in the first round. The places in the lower part are $e^B_{03}$ and $e^B_{22}$ respectively, where $e^B_{03}$ is in the second round and $e^B_{22}$ in the first round. These two checks can be performed concurrently since $e^T_{00}$ and $e^T_{20}$ are concurrent, as are $e^B_{22}$ and $e^B_{03}$. For the initial colour we have the extra initial round which does not get repeated. 
	
	The colouring $f$ is well-defined as the set of bad markings $\badmarkings_{\mathit{Same}}$ ensures that every two colours chosen when both top rounds and both bottom rounds are the same the colours chosen also must be the same in a winning strategy. This can be seen with the help of Fig. \ref{fig-concurrency-table}, too. There, we can find a sequence of checks throughout one round such that the colour chosen at the beginning of that round (e.g. $e^B_{03}$) is still the same colour as chosen at the end of that round (e.g $e^B_{04}$) (assuming that the top round does not change too).
	%
	%
	%
	%
\end{proof}
It is easy to see that the Petri game presented here performs at most 3 checks in every possible sequence of transitions. Therefore this Petri game has at most 3 system players. For the undecidability result it is sufficient to have 2 concurrent checks. Therefore, the number of system players can be further reduced to 2 by adapting the Petri game. This could be done by adding two environment players that are consumed performing a check.
Furthermore, by adding a transition for each bad marking to a bad place, the set of bad markings can be simplified to a set of bad places, the local safety condition.

%% file: sections/5-related-work-and-conclusions.tex
%

Our contribution to the synthesis problem of distributed systems is that this problem is  undecidable for 6-player Petri games under a local safety condition and that two system processes are sufficient to obtain undecidability. This adds neatly to the result in \cite{DBLP:journals/corr/abs-1710-05368}, where Petri games with one system player are solved in exponential time. Furthermore, the synthesis problem for 2-player Petri games belongs to the class of NP-complete problems, and we provide a non-deterministic exponential upper bound for this problem for Petri games with a variable number of players up to 4 under a global safety condition. This is a new class of Petri games for which the synthesis problem is decidable. 
The case of 5 players remains open.

%% file: example.bbl
\begin{thebibliography}{10}
\providecommand{\bibitemdeclare}[2]{}
\providecommand{\surnamestart}{}
\providecommand{\surnameend}{}
\providecommand{\urlprefix}{Available at }
\providecommand{\url}[1]{\texttt{#1}}
\providecommand{\href}[2]{\texttt{#2}}
\providecommand{\urlalt}[2]{\href{#1}{#2}}
\providecommand{\doi}[1]{doi:\urlalt{https://doi.org/#1}{#1}}
\providecommand{\eprint}[1]{arXiv:\urlalt{https://arxiv.org/abs/#1}{#1}}
\providecommand{\bibinfo}[2]{#2}

\bibitemdeclare{inproceedings}{DBLP:conf/concur/BeutnerFH19}
\bibitem{DBLP:conf/concur/BeutnerFH19}
\bibinfo{author}{Raven \surnamestart Beutner\surnameend},
  \bibinfo{author}{Bernd \surnamestart Finkbeiner\surnameend} \&
  \bibinfo{author}{Jesko \surnamestart Hecking{-}Harbusch\surnameend}
  (\bibinfo{year}{2019}): \emph{\bibinfo{title}{Translating Asynchronous Games
  for Distributed Synthesis}}.
\newblock In \bibinfo{editor}{Wan~J. \surnamestart Fokkink\surnameend} \&
  \bibinfo{editor}{Rob \surnamestart van Glabbeek\surnameend}, editors:
  {\slshape \bibinfo{booktitle}{30th International Conference on Concurrency
  Theory, {CONCUR} 2019}}, {\slshape \bibinfo{series}{LIPIcs}}
  \bibinfo{volume}{140}, \bibinfo{publisher}{Schloss Dagstuhl - Leibniz-Zentrum
  f{\"{u}}r Informatik}, pp. \bibinfo{pages}{26:1--26:16},
  \doi{10.4230/LIPIcs.CONCUR.2019.26}.

\bibitemdeclare{inproceedings}{DBLP:journals/corr/abs-1711-10636}
\bibitem{DBLP:journals/corr/abs-1711-10636}
\bibinfo{author}{Roderick \surnamestart Bloem\surnameend},
  \bibinfo{author}{Sven \surnamestart Schewe\surnameend} \&
  \bibinfo{author}{Ayrat \surnamestart Khalimov\surnameend}
  (\bibinfo{year}{2017}): \emph{\bibinfo{title}{{CTL}* synthesis via {LTL}
  synthesis}}.
\newblock In \bibinfo{editor}{Dana \surnamestart Fisman\surnameend} \&
  \bibinfo{editor}{Swen \surnamestart Jacobs\surnameend}, editors: {\slshape
  \bibinfo{booktitle}{Proceedings Sixth Workshop on Synthesis, SYNT@CAV 2017}},
  {\slshape \bibinfo{series}{{EPTCS}}} \bibinfo{volume}{260}, pp.
  \bibinfo{pages}{4--22}, \doi{10.4204/EPTCS.260.4}.

\bibitemdeclare{article}{A.Church-reactive-systems}
\bibitem{A.Church-reactive-systems}
\bibinfo{author}{Alonzo \surnamestart Church\surnameend}
  (\bibinfo{year}{1957}): \emph{\bibinfo{title}{Applications of recursive
  arithmetic to the problem of circuit synthesis}}.
\newblock {\slshape \bibinfo{journal}{Summaries of the Summer Institute of
  Symbolic Logic}} \bibinfo{volume}{1}, pp. \bibinfo{pages}{3--50}.

\bibitemdeclare{article}{DBLP:journals/acta/Engelfriet91}
\bibitem{DBLP:journals/acta/Engelfriet91}
\bibinfo{author}{Joost \surnamestart Engelfriet\surnameend}
  (\bibinfo{year}{1991}): \emph{\bibinfo{title}{Branching Processes of {Petri}
  Nets}}.
\newblock {\slshape \bibinfo{journal}{Acta Inf.}}
  \bibinfo{volume}{28}(\bibinfo{number}{6}), pp. \bibinfo{pages}{575--591},
  \doi{10.1007/BF01463946}.

\bibitemdeclare{article}{DBLP:journals/fmsd/EsparzaRV02}
\bibitem{DBLP:journals/fmsd/EsparzaRV02}
\bibinfo{author}{Javier \surnamestart Esparza\surnameend},
  \bibinfo{author}{Stefan \surnamestart R{\"{o}}mer\surnameend} \&
  \bibinfo{author}{Walter \surnamestart Vogler\surnameend}
  (\bibinfo{year}{2002}): \emph{\bibinfo{title}{An Improvement of McMillan's
  Unfolding Algorithm}}.
\newblock {\slshape \bibinfo{journal}{Formal Methods Syst. Des.}}
  \bibinfo{volume}{20}(\bibinfo{number}{3}), pp. \bibinfo{pages}{285--310},
  \doi{10.1023/A:1014746130920}.

\bibitemdeclare{inproceedings}{DBLP:conf/birthday/Finkbeiner15}
\bibitem{DBLP:conf/birthday/Finkbeiner15}
\bibinfo{author}{Bernd \surnamestart Finkbeiner\surnameend}
  (\bibinfo{year}{2015}): \emph{\bibinfo{title}{Bounded Synthesis for {Petri}
  Games}}.
\newblock In \bibinfo{editor}{Roland \surnamestart Meyer\surnameend},
  \bibinfo{editor}{Andr{\'{e}} \surnamestart Platzer\surnameend} \&
  \bibinfo{editor}{Heike \surnamestart Wehrheim\surnameend}, editors: {\slshape
  \bibinfo{booktitle}{Correct System Design - Symposium in Honor of
  Ernst-R{\"{u}}diger Olderog on the Occasion of His 60th Birthday,
  Proceedings}}, {\slshape \bibinfo{series}{Lecture Notes in Computer Science}}
  \bibinfo{volume}{9360}, \bibinfo{publisher}{Springer}, pp.
  \bibinfo{pages}{223--237}, \doi{10.1007/978-3-319-23506-6\_15}.

\bibitemdeclare{inproceedings}{DBLP:journals/corr/abs-1710-05368}
\bibitem{DBLP:journals/corr/abs-1710-05368}
\bibinfo{author}{Bernd \surnamestart Finkbeiner\surnameend} \&
  \bibinfo{author}{Paul \surnamestart G{\"o}lz\surnameend}
  (\bibinfo{year}{2018}): \emph{\bibinfo{title}{{Synthesis in Distributed
  Environments}}}.
\newblock In \bibinfo{editor}{Satya \surnamestart Lokam\surnameend} \&
  \bibinfo{editor}{R.~\surnamestart Ramanujam\surnameend}, editors: {\slshape
  \bibinfo{booktitle}{37th IARCS Annual Conference on Foundations of Software
  Technology and Theoretical Computer Science (FSTTCS 2017)}}, {\slshape
  \bibinfo{series}{Leibniz International Proceedings in Informatics
  (LIPIcs)}}~\bibinfo{volume}{93}, \bibinfo{publisher}{Schloss
  Dagstuhl--Leibniz-Zentrum fuer Informatik}, \bibinfo{address}{Dagstuhl,
  Germany}, pp. \bibinfo{pages}{28:1--28:14},
  \doi{10.4230/LIPIcs.FSTTCS.2017.28}.

\bibitemdeclare{article}{FINKBEINER2017181}
\bibitem{FINKBEINER2017181}
\bibinfo{author}{Bernd \surnamestart Finkbeiner\surnameend} \&
  \bibinfo{author}{Ernst-Rüdiger \surnamestart Olderog\surnameend}
  (\bibinfo{year}{2017}): \emph{\bibinfo{title}{Petri games: Synthesis of
  distributed systems with causal memory}}.
\newblock {\slshape \bibinfo{journal}{Information and Computation}}
  \bibinfo{volume}{253}, pp. \bibinfo{pages}{181--203},
  \doi{10.1016/j.ic.2016.07.006}.
\newblock \bibinfo{note}{GandALF 2014}.

\bibitemdeclare{inproceedings}{synthesis-FinkbeinerS05}
\bibitem{synthesis-FinkbeinerS05}
\bibinfo{author}{Bernd \surnamestart Finkbeiner\surnameend} \&
  \bibinfo{author}{Sven \surnamestart Schewe\surnameend}
  (\bibinfo{year}{2005}): \emph{\bibinfo{title}{Uniform Distributed
  Synthesis}}.
\newblock In: {\slshape \bibinfo{booktitle}{20th {IEEE} Symposium on Logic in
  Computer Science {(LICS} 2005), Proceedings}}, pp. \bibinfo{pages}{321--330},
  \doi{10.1109/LICS.2005.53}.

\bibitemdeclare{inproceedings}{DBLP:conf/icalp/GenestGMW13}
\bibitem{DBLP:conf/icalp/GenestGMW13}
\bibinfo{author}{Blaise \surnamestart Genest\surnameend}, \bibinfo{author}{Hugo
  \surnamestart Gimbert\surnameend}, \bibinfo{author}{Anca \surnamestart
  Muscholl\surnameend} \& \bibinfo{author}{Igor \surnamestart
  Walukiewicz\surnameend} (\bibinfo{year}{2013}):
  \emph{\bibinfo{title}{Asynchronous Games over Tree Architectures}}.
\newblock In \bibinfo{editor}{Fedor~V. \surnamestart Fomin\surnameend},
  \bibinfo{editor}{Rusins \surnamestart Freivalds\surnameend},
  \bibinfo{editor}{Marta~Z. \surnamestart Kwiatkowska\surnameend} \&
  \bibinfo{editor}{David \surnamestart Peleg\surnameend}, editors: {\slshape
  \bibinfo{booktitle}{Automata, Languages, and Programming - 40th International
  Colloquium, {ICALP} 2013, Proceedings, Part {II}}}, {\slshape
  \bibinfo{series}{Lecture Notes in Computer Science}} \bibinfo{volume}{7966},
  \bibinfo{publisher}{Springer}, pp. \bibinfo{pages}{275--286},
  \doi{10.1007/978-3-642-39212-2\_26}.

\bibitemdeclare{inproceedings}{DBLP:conf/fsttcs/Gimbert17}
\bibitem{DBLP:conf/fsttcs/Gimbert17}
\bibinfo{author}{Hugo \surnamestart Gimbert\surnameend} (\bibinfo{year}{2017}):
  \emph{\bibinfo{title}{On the Control of Asynchronous Automata}}.
\newblock In \bibinfo{editor}{Satya~V. \surnamestart Lokam\surnameend} \&
  \bibinfo{editor}{R.~\surnamestart Ramanujam\surnameend}, editors: {\slshape
  \bibinfo{booktitle}{37th {IARCS} Annual Conference on Foundations of Software
  Technology and Theoretical Computer Science, {FSTTCS} 2017, India}},
  {\slshape \bibinfo{series}{LIPIcs}}~\bibinfo{volume}{93},
  \bibinfo{publisher}{Schloss Dagstuhl - Leibniz-Zentrum f{\"{u}}r Informatik},
  pp. \bibinfo{pages}{30:1--30:15}, \doi{10.4230/LIPIcs.FSTTCS.2017.30}.

\bibitemdeclare{article}{DBLP:journals/lmcs/Gimbert22}
\bibitem{DBLP:journals/lmcs/Gimbert22}
\bibinfo{author}{Hugo \surnamestart Gimbert\surnameend} (\bibinfo{year}{2022}):
  \emph{\bibinfo{title}{Distributed Asynchronous Games With Causal Memory are
  Undecidable}}.
\newblock {\slshape \bibinfo{journal}{Log. Methods Comput. Sci.}}
  \bibinfo{volume}{18}(\bibinfo{number}{3}), \doi{10.46298/lmcs-18(3:30)2022}.

\bibitemdeclare{inproceedings}{DBLP:conf/apn/HannibalO22}
\bibitem{DBLP:conf/apn/HannibalO22}
\bibinfo{author}{Paul \surnamestart Hannibal\surnameend} \&
  \bibinfo{author}{Ernst{-}R{\"{u}}diger \surnamestart Olderog\surnameend}
  (\bibinfo{year}{2022}): \emph{\bibinfo{title}{The Synthesis Problem for
  Repeatedly Communicating Petri Games}}.
\newblock In \bibinfo{editor}{Luca \surnamestart Bernardinello\surnameend} \&
  \bibinfo{editor}{Laure \surnamestart Petrucci\surnameend}, editors: {\slshape
  \bibinfo{booktitle}{Application and Theory of Petri Nets and Concurrency -
  43rd International Conference, {PETRI} {NETS} 2022, Proceedings}}, {\slshape
  \bibinfo{series}{Lecture Notes in Computer Science}} \bibinfo{volume}{13288},
  \bibinfo{publisher}{Springer}, pp. \bibinfo{pages}{236--257},
  \doi{10.1007/978-3-031-06653-5\_13}.

\bibitemdeclare{article}{Kupferman2001389}
\bibitem{Kupferman2001389}
\bibinfo{author}{O.~\surnamestart Kupferman\surnameend} \&
  \bibinfo{author}{M.Y. \surnamestart Vardi\surnameend} (\bibinfo{year}{2001}):
  \emph{\bibinfo{title}{Synthesizing distributed systems}}.
\newblock {\slshape \bibinfo{journal}{Proceedings - Symposium on Logic in
  Computer Science}}, pp. \bibinfo{pages}{389--398},
  \doi{10.1109/LICS.2001.932514}.

\bibitemdeclare{inproceedings}{DBLP:conf/concur/MadhusudanT02}
\bibitem{DBLP:conf/concur/MadhusudanT02}
\bibinfo{author}{P.~\surnamestart Madhusudan\surnameend} \&
  \bibinfo{author}{P.~S. \surnamestart Thiagarajan\surnameend}
  (\bibinfo{year}{2002}): \emph{\bibinfo{title}{A Decidable Class of
  Asynchronous Distributed Controllers}}.
\newblock In \bibinfo{editor}{Lubos \surnamestart Brim\surnameend},
  \bibinfo{editor}{Petr \surnamestart Jancar\surnameend},
  \bibinfo{editor}{Mojm{\'{\i}}r \surnamestart
  Kret{\'{\i}}nsk{\'{y}}\surnameend} \& \bibinfo{editor}{Anton{\'{\i}}n
  \surnamestart Kucera\surnameend}, editors: {\slshape
  \bibinfo{booktitle}{{CONCUR} 2002 - Concurrency Theory, 13th International
  Conference, Proceedings}}, {\slshape \bibinfo{series}{Lecture Notes in
  Computer Science}} \bibinfo{volume}{2421}, \bibinfo{publisher}{Springer}, pp.
  \bibinfo{pages}{145--160}, \doi{10.1007/3-540-45694-5\_11}.

\bibitemdeclare{inproceedings}{DBLP:conf/fsttcs/MadhusudanTY05}
\bibitem{DBLP:conf/fsttcs/MadhusudanTY05}
\bibinfo{author}{P.~\surnamestart Madhusudan\surnameend},
  \bibinfo{author}{P.~S. \surnamestart Thiagarajan\surnameend} \&
  \bibinfo{author}{Shaofa \surnamestart Yang\surnameend}
  (\bibinfo{year}{2005}): \emph{\bibinfo{title}{The {MSO} Theory of Connectedly
  Communicating Processes}}.
\newblock In \bibinfo{editor}{Ramaswamy \surnamestart Ramanujam\surnameend} \&
  \bibinfo{editor}{Sandeep \surnamestart Sen\surnameend}, editors: {\slshape
  \bibinfo{booktitle}{{FSTTCS} 2005: Foundations of Software Technology and
  Theoretical Computer Science, 25th International Conference, Proceedings}},
  {\slshape \bibinfo{series}{Lecture Notes in Computer Science}}
  \bibinfo{volume}{3821}, \bibinfo{publisher}{Springer}, pp.
  \bibinfo{pages}{201--212}, \doi{10.1007/11590156\_16}.

\bibitemdeclare{inproceedings}{DBLP:conf/icalp/Muscholl15}
\bibitem{DBLP:conf/icalp/Muscholl15}
\bibinfo{author}{Anca \surnamestart Muscholl\surnameend}
  (\bibinfo{year}{2015}): \emph{\bibinfo{title}{Automated Synthesis of
  Distributed Controllers}}.
\newblock In \bibinfo{editor}{Magn{\'{u}}s~M. \surnamestart
  Halld{\'{o}}rsson\surnameend}, \bibinfo{editor}{Kazuo \surnamestart
  Iwama\surnameend}, \bibinfo{editor}{Naoki \surnamestart Kobayashi\surnameend}
  \& \bibinfo{editor}{Bettina \surnamestart Speckmann\surnameend}, editors:
  {\slshape \bibinfo{booktitle}{Automata, Languages, and Programming - 42nd
  International Colloquium, {ICALP} 2015, Proceedings, Part {II}}}, {\slshape
  \bibinfo{series}{Lecture Notes in Computer Science}} \bibinfo{volume}{9135},
  \bibinfo{publisher}{Springer}, pp. \bibinfo{pages}{11--27},
  \doi{10.1007/978-3-662-47666-6\_2}.

\bibitemdeclare{inproceedings}{DBLP:conf/icalp/PnueliR89}
\bibitem{DBLP:conf/icalp/PnueliR89}
\bibinfo{author}{Amir \surnamestart Pnueli\surnameend} \& \bibinfo{author}{Roni
  \surnamestart Rosner\surnameend} (\bibinfo{year}{1989}):
  \emph{\bibinfo{title}{On the Synthesis of an Asynchronous Reactive Module}}.
\newblock In \bibinfo{editor}{Giorgio \surnamestart Ausiello\surnameend},
  \bibinfo{editor}{Mariangiola \surnamestart Dezani{-}Ciancaglini\surnameend}
  \& \bibinfo{editor}{Simona Ronchi~Della \surnamestart Rocca\surnameend},
  editors: {\slshape \bibinfo{booktitle}{Automata, Languages and Programming,
  16th International Colloquium, ICALP89, Proceedings}}, {\slshape
  \bibinfo{series}{Lecture Notes in Computer Science}} \bibinfo{volume}{372},
  \bibinfo{publisher}{Springer}, pp. \bibinfo{pages}{652--671},
  \doi{10.1007/BFb0035790}.

\bibitemdeclare{inproceedings}{DBLP:conf/focs/PnueliR90}
\bibitem{DBLP:conf/focs/PnueliR90}
\bibinfo{author}{Amir \surnamestart Pnueli\surnameend} \& \bibinfo{author}{Roni
  \surnamestart Rosner\surnameend} (\bibinfo{year}{1990}):
  \emph{\bibinfo{title}{Distributed Reactive Systems Are Hard to Synthesize}}.
\newblock In: {\slshape \bibinfo{booktitle}{31st Annual Symposium on
  Foundations of Computer Science, Volume {II}}}, pp.
  \bibinfo{pages}{746--757}, \doi{10.1109/FSCS.1990.89597}.

\bibitemdeclare{phdthesis}{Rosnerpipelines}
\bibitem{Rosnerpipelines}
\bibinfo{author}{R.~\surnamestart Rosner\surnameend} (\bibinfo{year}{1992}):
  \emph{\bibinfo{title}{Modular synthesis of reactive systems}}.
\newblock Ph.D. thesis, \bibinfo{school}{Weizmann Institute of Science,
  Rehovot, Israel}.

\end{thebibliography}
